\documentclass[letterpaper]{article} 
\usepackage[draft]{aaai25}  
\usepackage{times}  
\usepackage{helvet}  
\usepackage{courier}  
\usepackage[hyphens]{url}  
\usepackage{graphicx} 
\usepackage{natbib}  
\usepackage{caption} 
\usepackage{algorithm}
\usepackage{algorithmic}
\usepackage{subcaption}
\usepackage{booktabs}

\usepackage{newfloat}
\usepackage{listings}
\DeclareCaptionStyle{ruled}{labelfont=normalfont,labelsep=colon,strut=off} 
\floatstyle{ruled}
\newfloat{listing}{tb}{lst}{}
\floatname{listing}{Listing}
\usepackage{amsmath}
\usepackage{amsthm}
\usepackage{amsfonts}
\usepackage{mathtools}

\newcommand{\CI}{\mathrel{\text{\scalebox{1.07}{$\perp\mkern-10mu\perp$}}}}
\newcommand{\otilde}{\widetilde{\oplus}}
\newcommand{\bigotilde}{\widetilde{\bigoplus}}
\newcommand{\tilmu}{\widetilde{\boldsymbol{\mu}}}
\newcommand{\localmin}[1]{\mathrm{local}\text{ }\mathrm{minima}\left(#1\right)}
\newcommand{\globmin}[1]{\mathrm{global}\text{ }\mathrm{minima}\left(#1\right)}

\newcommand{\algiterate}[1]{{\left(#1\right)}}
\newcommand{\pderiv}[1]{\frac{\partial}{\partial #1}}

\newtheorem{problem}{Problem}
\newtheorem{theorem}{Theorem}
\newtheorem{lemma}{Lemma}
\newtheorem{corollary}{Corollary}

\newtheorem{repeatedtheorem}{Theorem}

\newtheorem{innercustomthm}{Theorem}
\newenvironment{customthm}[1]
  {\renewcommand\theinnercustomthm{#1}\innercustomthm}
  {\endinnercustomthm}

\allowdisplaybreaks

\title{Graphical Modelling without Independence Assumptions for Uncentered Data}
\author {
    Bailey Andrew\textsuperscript{\rm 1},
    David R Westhead\textsuperscript{\rm 1},
    Luisa Cutillo\textsuperscript{\rm 1}
}
\affiliations {
    \textsuperscript{\rm 1}University of Leeds\\
    sceba@leeds.ac.uk,
    D.R.Westhead@leeds.ac.uk,
    L.Cutillo@leeds.ac.uk
}

\begin{document}

\maketitle

\begin{abstract}
The independence assumption is a useful tool to increase the tractability of one's modelling framework.  However, this assumption does not match reality; failing to take dependencies into account can cause models to fail dramatically.  The field of multi-axis graphical modelling (also called multi-way modelling, Kronecker-separable modelling) has seen growth over the past decade, but these models require that the data have zero mean.  In the multi-axis case, inference is typically done in the single sample scenario, making mean inference impossible.

In this paper, we demonstrate how the zero-mean assumption can cause egregious modelling errors, as well as propose a relaxation to the zero-mean assumption that allows the avoidance of such errors.  Specifically, we propose the ``Kronecker-sum-structured mean'' assumption, which leads to models with nonconvex-but-unimodal log-likelihoods that can be solved efficiently with coordinate descent.
\end{abstract}

%

\section{Introduction}

We often wish to find networks (`graphs') that describe our data.  For example, we may be interested in gene regulatory networks in biology, or social interaction networks in epidemiology.  In these cases, the graph itself is an object of interest.  When the goal of an analysis does not involve a graph, the creation of one can still be useful as a preprocessing step to further insights, such as for clustering.  Their use is not limited to clustering; the first step of the popular dimensionality reduction method UMAP, for example, is to create a (typically nearest neighbors) graph \cite{mcinnes_umap_2020}, but one could also experiment with other graphs.

While there are many types of graphs, this paper will focus on \textit{conditional dependency graphs}.  Intuitively, two vertices are connected in such a graph if they are still statistically dependent after conditioning out the other variables.  We denote the property that $x$ and $y$ are conditionally independent, given z, as $x \CI y \mid z$.  We will focus on the case where we condition over all other variables in the dataset $\mathcal{D}$, i.e. where $z = \mathcal{D}_{\backslash x\backslash y} \overset{\textit{def}}{=}\mathcal{D} - \{x, y\}$.

\begin{align*}
    x \CI y \mid z &\overset{\textit{def}}{\iff} \mathbb{P}\left[x, y \mid z\right] = \mathbb{P}\left[x \mid z\right]\mathbb{P}\left[y \mid z\right]
\end{align*}

Conditional dependency graphs have several convenient properties.  They are interpretable, and intuitively capture the `direct effect' of two variables on each other (rather than `indirect effects' passing through confounders and mediators).  They are typically sparse, and are related to causality\footnote{Under modest assumptions, that can be violated in some real-world scenarios, conditional dependency graphs form the `skeleton' of a causal directed acyclic graph, i.e. the causal graph with directionality removed.}.  They also have a very useful property in multivariate Gaussian datasets, namely that two Gaussian variables are conditionally independent if and only if the corresponding term in the inverse covariance matrix is zero.

\begin{align*}
    x \CI y \mid \mathcal{D}_{\backslash x\backslash y} &\overset{\textit{Gaussian}}{\iff} \mathbf{\Sigma}^{-1}_{xy} = 0
\end{align*}

The inverse covariance matrix is also called the precision matrix, and we will denote it $\mathbf{\Psi}$.  It can be interpreted as an adjacency matrix for the graph of conditional dependencies.  Due to this convenient correspondence, and due to the fact that practically all multi-axis work has so far been limited to the Gaussian case, we will assume in this paper that our data is Gaussian.  Through tools such as the Nonparanormal Skeptic \cite{liu_nonparanormal_2012}, this assumption can be weakend to the Gaussian copula assumption (intuitively, the variables have arbitrary continuous marginals but still interact `Gaussian-ly').

When making independence assumptions, this problem is solved by well-known methods such as the Graphical Lasso \cite{friedman_sparse_2008}.  However, independence assumptions are often false in practice.  They become particularly egregious in the case of omics data, such as single cell RNA-sequencing (scRNA-seq).  These datasets take the form of a cells-by-genes matrix.  We might be interested in a gene regulatory network, in which case we assume the cells are independent.  Alternatively, we might be interested in the cellular microenvironment, in which case we assume the genes are independent.  Whichever one we want to learn, we are put in the awkward position of assuming the other does not exist.

To avoid independence assumptions altogether, we can vectorize our dataset.  Instead of $n$ samples of $m$ features, we have $1$ sample of $nm$ features.  For scRNA-seq, the precision matrix is of size $O(n^2m^2)$ and represents the dependencies between different (cell, gene) pairs.  The problem becomes both computationally and statistically intractable; we quickly run out of space to store such matrices, and have no way of being confident in our results to any degree of statistical certainty.  We need to find a middle ground between independence and full dependence.

To do this, we can take advantage of matrix structure.  We are not truly interested in a graph of (cell, gene) pairs, but rather a graph of cells and a graph of genes.  Thus, we can make a `graph decomposition assumption'; our (cell, gene) pair graph should be able to be factored (for some definition of factoring) into the cell graph and the gene graph.

Several choices of decomposition have appeared in the literature, such as the Kronecker product \cite{dutilleul_mle_1999} and squared Kronecker sum decompositions \cite{wang_sylvester_2020}; in this paper, we focus on the Kronecker sum decomposition \cite{kalaitzis_bigraphical_2013} due to its popularity, convexity, relationship to conditional dependence, interpretability as a Cartesian product, and correspondence to a maximum entropy distribution.  The Kronecker sum is denoted $\oplus$ and is defined as $\mathbf{A}_{a\times a} \oplus \mathbf{B}_{b\times b} = \mathbf{A}_{a\times a} \otimes \mathbf{I}_{b\times b} + \mathbf{I}_{a\times a} \otimes \mathbf{B}_{b\times b}$, where $\otimes$ is the Kronecker product.  The operation is associative, allowing a straightforward generalization to more than two axes (i.e. tensor-variate datasets); this was first explored by \citeauthor{greenewald_tensor_2019} (\citeyear{greenewald_tensor_2019}).  Under the zero-mean and Kronecker sum decomposability assumptions, the model for a dataset $\mathcal{D}$ is as follows:

\begin{align*}
    \mathcal{D} &\sim \mathcal{N}_{KS}\left(\mathbf{0}, \{\mathbf{\Psi}_\ell\}_\ell\right) \\
    \iff \mathrm{vec}\left[\mathcal{D}\right] &\sim \mathcal{N}\left(\mathbf{0},\left(\bigoplus_\ell \mathbf{\Psi}_\ell\right)^{-1}\right)
\end{align*}

The zero-mean assumption is conspicuous, but necessary; recall that we transformed our dataset from $n$ samples and $m$ features to a single sample of $nm$ features.  Estimating the mean of a single sample is a recipe for disaster; when one centers the dataset, they will be left with a constant zero vector.

We could consider decomposing the mean of our model, just like for our precision matrices.  A natural and analogous decomposition would be that our mean takes the form $\boldsymbol{\mu}_\mathrm{cells} \otilde \boldsymbol{\mu}_\mathrm{genes} = \boldsymbol{\mu}_\mathrm{cells} \otimes \mathbf{1} + \mathbf{1} \otimes \boldsymbol{\mu}_\mathrm{genes}$; this operation intermixes the means in the same way the Kronecker sum intermixes the variances.  For Kronecker \textit{product} models (not Kronecker sum models), this has been considered before \cite{allen_transposable_2010}.  It was shown that the maximum likelihood for $\mu_\ell$ has the following form (in the two-axis case, where $d_\ell$ represents the number of elements in axis $\ell$):

\begin{align*}
    \boldsymbol{\mu}_1 &= \frac{\mathbf{1}^T\left(\mathbf{X} - \mathbf{1}\boldsymbol{\mu}_2^T\right)}{d_2} & \boldsymbol{\mu}_2 &= \frac{\mathbf{1}^T\left(\mathbf{X} - \mathbf{1}\boldsymbol{\mu}_1^T\right)}{d_1}
\end{align*}

We will see later that such formulas are not as simple in the Kronecker sum case.  First, though, let us note that the estimate of $\boldsymbol{\mu}_\ell$ does not depend on the dependencies in the data ($\mathbf{\Psi}_\ell$).  This is a mathematically convenient property to have, but is philosophically troubling.  Suppose we have a dataset of only two points, both distinct and independent.  The true mean value is just their average.  Now, suppose we duplicate the first point several times, perhaps with some modest noise added - the mean estimate will lie much closer to the first point.

In fact, we can choose how many times to duplicate the first or second points in order to place the estimated mean anywhere between the original points.  However, this estimated mean is an artifact of the dependencies in our data!  The true mean, ultimately, still lies at its original location.

As an example, suppose we were studying a (non-Hodgkin's) lymphoma dataset that was representative of the general population.  Only 2\% of lymphoma cases are lymphoblastic lymphoma\footnote{According to the Leukemia Foundation 
charity.}, but specialized CAR-T cell treatments have been approved for it in the UK, such as Brexucabtagene Autoleucel 
(NICE, 2013).  
In models that don't take into account dependencies, this signal could be drowned out by the other, more common, subtypes; the `average lymphoma case' would look much more like the average case of its largest subtype, rather than being representative of all subtypes.  The fact that rare subtypes may benefit from specialized treatments simply gets drowned out by the properties of the others.

\subsection{Our Contributions}

\begin{figure}[h!]
    \centering
    \includegraphics[width=0.7\linewidth]{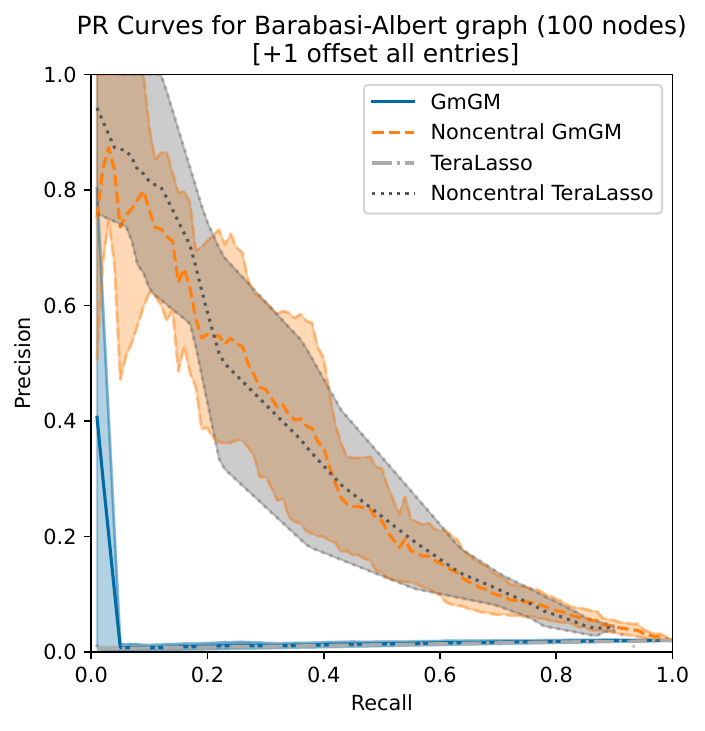}
    \caption{Precision and recall for on a synthetic dataset in which data was generated from a Kronecker-sum normal distribution, and then had a constant value of 1 added to every entry of the matrix.
    Error bars are the best/worst performance over 10 trials; the center line is average performance.}
    \label{fig:ba-total-offset}
\end{figure}

We are naturally led to the following question:

\begin{problem}
    \label{prob:estimability}
    For the noncentral Kronecker sum model, does using the mean decomposition $\bigotilde_\ell \boldsymbol{\mu}_\ell$ allow its parameters to be effectively estimated?
\end{problem}

In this paper, we show that the answer to this question is yes.  For the Kronecker-sum-structured model, the estimators for $\boldsymbol{\mu}_\ell$ depend on $\mathbf{\Psi}_\ell$.  However, due to this dependency, our problem is no longer convex; there could be some non-global minima!

\begin{problem}
    \label{prob:mle}
    Is our estimator for the parameters of the noncentral Kronecker-sum-structured normal distribution guaranteed to be the \textit{global} maximum likelihood estimator?
\end{problem}

We also show that the answer to this question is also affirmative; in fact, there are no non-global minima, despite the nonconvexity.

One may also wonder if the zero-mean assumption is truly problematic.  In fact, the zero-mean assumption can have an extremely dramatic negative effect on performance.  In the extremely simple scenario in which we add a constant offset to every element, we can see that methods that do not try to estimate the mean are practically worthless, whereas those that do (which we shall call `noncentral') maintain reasonable performance (Figure \ref{fig:ba-total-offset}).

\subsection{Notation}

In this paper, lowercase letters $a$ will refer to scalars, lowercase bold $\mathbf{v}$ will refer to vectors, uppercase bold $\mathbf{M}$ will refer to matrices, and uppercase calligraphic $\mathcal{T}$ will refer to tensors.  We will conceptualize our dataset to be a tensor $\mathcal{D}$ with axes lengths $d_\ell$ for each axis $\ell$.  $d_{<\ell}$ and $d_{>\ell}$ are shorthands for the product of lengths of axes before and after $\ell$, respectively.  $d_{\backslash\ell}$ is the product of all axes other than $\ell$; from the perspective of $\ell$, this is the number of samples for that axis.  $d_\forall$ is the product of all $d_\ell$.  We typically follow the notation of \citeauthor{kolda_tensor_2009} (\citeyear{kolda_tensor_2009}) for tensor operations.

Let $\mathbf{I}_a$ and $\mathbf{1}_a$ be the $a\times a$ identity matrix and length $a$ vector of ones, respectively.  We formally define the Kronecker sum of matrices and analogous operation on vectors as follows:

\begin{align*}
    \bigoplus_\ell \mathbf{\Psi}_\ell &= \sum_\ell \mathbf{I}_{d_{<\ell}} \otimes \mathbf{\Psi}_\ell \otimes \mathbf{I}_{d_{>\ell}} \\
    \bigotilde_\ell \boldsymbol{\mu}_\ell &= \sum_\ell \mathbf{1}_{d_{<\ell}} \otimes \boldsymbol{\mu}_\ell \otimes \mathbf{1}_{d_{>\ell}}
\end{align*}

We will let $\mathbf{\Psi}_{\backslash\ell}$ be a shorthand for $\bigoplus_{\ell'\neq\ell}\mathbf{\Psi}_{\ell'}$, with an analogous definition for $\boldsymbol{\mu}_{\backslash\ell}$.  Kronecker products, and by extension Kronecker sums, have a convenient permutation-invariance property; we can swap the order of summands without affecting the results, as long as we perform this swap consistently across an equation and permute matrices accordingly; we can typically rewrite equations involving $\bigoplus_\ell \mathbf{\Psi}_\ell$ into those involving $\mathbf{\Psi}_\ell \oplus \mathbf{\Psi}_{\backslash\ell}$.  We will define $\mathbf{\Omega} = \bigoplus_\ell \mathbf{\Psi}_\ell$ and $\boldsymbol{\omega} = \bigotilde_\ell \boldsymbol{\mu}_\ell$.  We define noncentral Kronecker sum normal distribution as:

\begin{align*}
    \mathrm{vec}\left[\mathcal{D}\right] &\sim \mathcal{N}\left(\boldsymbol{\omega}, \boldsymbol{\Omega}^{-1}\right)
\end{align*}

We let $\mathbb{K}_\mathbf{M}$ be the space of Kronecker-sum-decomposable matrices\footnote{Technically, we need to specify the dimensionality of each term in the decomposition to define this space, but this would be notationally cumbersome and is always clear from context.}, i.e. $\mathbf{\Omega} \in \mathbb{K}_\mathbf{M}$.  Likewise, $\boldsymbol{\omega} \in \mathbb{K}_\mathbf{v}$ is the space of Kronecker-sum-decomposable vectors.  Finally, let $\mathbb{K}_{\mathbf{M}\mathbf{v}}$ be the space of vectors of the form $\mathbf{X}\mathbf{y}$, where $\mathbf{X} \in \mathbb{K}_\mathbf{M}, \mathbf{y} \in \mathbb{K}_\mathbf{v}$.  All of these spaces are linear subspaces of $\mathbb{R}^{d_\forall}$ or $\mathbb{R}^{d_\forall\times d_\forall}$.  We may use $\mathbb{K}$ as a shorthand when discussing properties that apply to both $\mathbb{K}_\mathbf{M}$ and $\mathbb{K}_\mathbf{v}$.

It is important to note that the parameterization of $\mathbb{K}$ we have used so far is not identifiable.  If $\sum_\ell c_\ell = 0$, then $\bigoplus_\ell \left(\mathbf{\Psi}_\ell + c_\ell\mathbf{I}\right) = \bigoplus_\ell \mathbf{\Psi}_\ell$.  \citeauthor{greenewald_tensor_2019} (\citeyear{greenewald_tensor_2019}) use the identifiable representation $\tau\mathbf{I}_\forall + \bigoplus_\ell \widetilde{\mathbf{\Psi}}_\ell$, where $\mathrm{tr}\left[\widetilde{\mathbf{\Psi}}_\ell\right] = 0$.  Likewise, we can identifiably represent our mean decomposition as $m\mathbf{1} + \bigotilde_\ell \widetilde{\boldsymbol{\mu}}_\ell$, where $\mathbf{1}^T\widetilde{\boldsymbol{\mu}}_\ell = 0$.

The last important property we will introduce is the matrix representation of Kronecker products of vectors: note that $\boldsymbol{\mu}_\ell \otimes \mathbf{1}_{d_{\backslash\ell}} = \left(\mathbf{I}_{d_\ell} \otimes \mathbf{1}_{d_{\backslash_\ell}}\right)\boldsymbol{\mu_\ell}$.

\section{Estimation Method}

As a shorthand, let $\mathbf{x} = \mathrm{vec}\left[\mathcal{D}\right]$.  Our negative log likelihood looks just like that of our normal distribution, except that the parameters are restricted to lie in $\mathbb{K}$.  To ensure existence of $\mathbf{\Omega}$, it is necessary to either add a regularization penalty, as in TeraLasso, or restrict $\mathbf{\Omega}$ to a low-rank subspace, as proposed by \citeauthor{andrew_making_2024} (\citeyear{andrew_making_2024}).  These are details of the estimator for $\mathbf{\Omega}$, which do not affect the estimator of $\boldsymbol{\omega}$; we do not focus on these aspects of the problem.

\begin{align*}
    \mathrm{p}\left(\mathcal{D}\right) &= \frac{\sqrt{\left|\mathbf{\Omega}\right|}}{\left(2\pi\right)^{\frac{d_\forall}{2}}}e^{\frac{-1}{2}\left(\boldsymbol{x} - \boldsymbol{\omega}\right)^T\mathbf{\Omega}\left(\boldsymbol{x} - \boldsymbol{\omega}\right)} \\
    \mathrm{NLL}\left(\mathcal{D}\right) &\propto \frac{-1}{2}\log\left|\mathbf{\Omega}\right| + \frac{1}{2}\left(\boldsymbol{x} - \boldsymbol{\omega}\right)^T\mathbf{\Omega}\left(\boldsymbol{x} - \boldsymbol{\omega}\right)
\end{align*}

Despite its ubiquity, the normal distribution is actually fairly poorly behaved from an optimization perspective; its NLL is not convex, nor is it even geodesically convex \cite{hosseini_matrix_2015}.  Thankfully, the MLE of the mean (for the unrestricted normal) has a closed form solution, $\mathbf{\omega} = \frac{1}{n}\sum_i^n \mathbf{x}_i$, so this nonconvexity does not affect results in practice.  Unfortunately, due to our $\mathbb{K}_\mathbf{v}$ restriction, this solution no longer holds.

\begin{figure}
    \centering
    \includegraphics[width=0.8\linewidth]{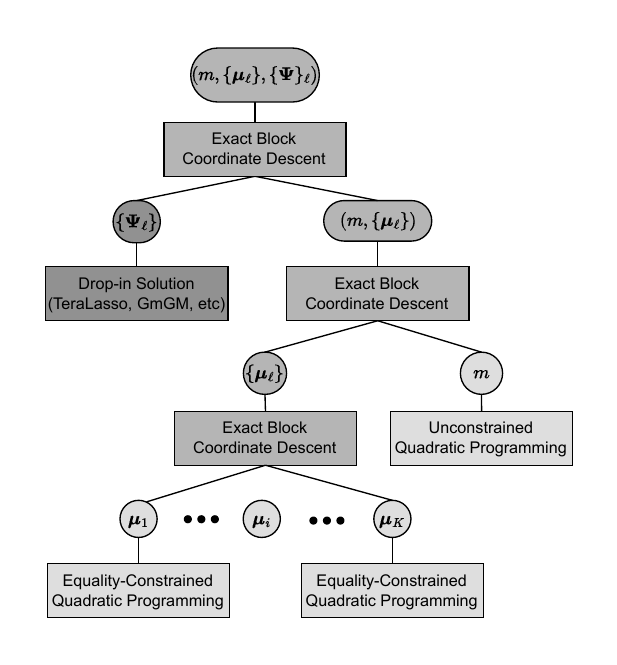}
    \caption{A graphical summary of the manner in which we divide our problem into sub-problems.  While the main problem is nonconvex, all sub-problems are convex, and by Theorem \ref{thm:unimodality} the main problem has a unique solution.}
    \label{fig:estimation-walkthrough}
\end{figure}

While it is not convex in $(\mathbf{\Omega}, \boldsymbol{\omega})$ jointly, it is convex in each argument individually - i.e, it is biconvex.  This suggests a flip-flop optimization scheme.  For fixed $\boldsymbol{\omega}$, the estimation of $\mathbf{\Omega}$ reduces to the noncentral case; several algorithms already exist to solve this.  Thus, we focus on optimization when $\mathbf{\Omega}$ is fixed.

Optimization w.r.t. $\boldsymbol{\omega}$ can be framed as a convex constrained quadratic programming (QP) problem.  However, it will turn out to be convenient to further break the problem up into its constituent parts $\left(m, \left\{\tilmu_\ell\right\}\right)$, and minimize them each individually.  Let $\theta_\ell = \mathbf{1}_{d_{\ell'}}^T\mathbf{\Psi}_{\ell'}\mathbf{1}_{d_{\ell'}}$, $\theta_{\backslash\ell} = \sum_{\ell'\neq\ell}\frac{d_\forall}{d_\ell d_{\ell'}} \theta_\ell$, and $\mathbf{x}_\ell$ be the vectorization of $\mathcal{D}$ done in the order as if axis $\ell$ were the first axis.  To preserve space, we will defer the algebraic manipulations to the appendix, and simply note here that the optimal value of $m$, with all other values fixed, is:

\begin{align}
    m &= \frac{\left(\mathbf{x} - \bigotilde_\ell \tilmu_\ell\right)\left(\bigotilde\mathbf{\Psi}_\ell\mathbf{1}_{d_\ell}\right)}{\sum_\ell d_{\backslash\ell}\theta_\ell} \label{eq:m}
\end{align}

Observe that, when the elements of $\mathbf{x}$ are independent, this formula expresses the average distance of the dataset from 0 (after accounting for the axis-wise means $\tilmu_\ell$).

For $\tilmu_\ell$ (with $\tilmu_{\backslash\ell}$ fixed, and again deferring the derivation to the appendix) we have the following QP problem:

\begin{align*}
    \underset{\tilmu_\ell}{\mathrm{argmin}} &\hspace{5pt} \tilmu_\ell^T \mathbf{A}_\ell \tilmu_\ell + \mathbf{b}_\ell^T\tilmu_\ell \\
    \mathrm{where} &\hspace{3pt} \mathbf{A}_\ell = d_{\backslash\ell}\mathbf{\Psi}_\ell + \theta_{\backslash\ell}\mathbf{I}_{d_\ell} \\
    &\hspace{3pt}\boldsymbol{b}_\ell = md_{\backslash\ell}\mathbf{1}_{d_\ell}^T\mathbf{\Psi}_\ell + m\theta_{\backslash\ell}\mathbf{1}_{d_\ell} + \left[\tilmu_{\backslash\ell}^T\mathbf{\Psi}_{\backslash\ell}\mathbf{1}_{d_{\backslash\ell}}\right]\mathbf{1}_{d_\ell} \\
    &\hspace{13pt}-\mathbf{x}_\ell^T\left(\mathbf{\Psi}_\ell \otimes \mathbf{1}_{d_{\backslash\ell}}\right) - \mathbf{x}_\ell^T\left(\mathbf{1}_\ell\otimes\mathbf{\Psi}_{\backslash\ell}\mathbf{1}_{d_{\backslash\ell}}\right) \\
    &\hspace{3pt} \tilmu_\ell^T\mathbf{1}_{d_\ell} = \mathbf{0}
\end{align*}

Parts of the definition of $\mathbf{b}_\ell$ can be simplified further, leading to more efficient computations, but such manipulations are notationally complicated to express; we defer them to the appendix, where we also show that $\mathbf{A}_\ell$ is guaranteed to be invertible.  QP problems with linear equality constraints have closed-form solutions.  In particular, we have that:

\begin{align}
    \tilmu_\ell = \frac{\boldsymbol{1}_{d_\ell}^T\boldsymbol{A}_\ell^{-1}\boldsymbol{b}_\ell}{\boldsymbol{1}_{d_\ell}^T\mathbf{A}_\ell^{-1}\boldsymbol{1}_{d_\ell}}\boldsymbol{A}_\ell^{-1}\mathbf{1}_{d_\ell} - \mathbf{A}_\ell^{-1}\boldsymbol{b}_\ell \label{eq:tilmu}
\end{align}

As we have closed-form solutions to all our coordinatewise minima, block coordinate descent is a natural solution.  While typically it is advisable to avoid directly calculating inverses $\mathbf{A}_\ell^{-1}$ in the solution to QP problems, inversion here is likely to be cheap; most solvers for $\mathbf{\Psi}_\ell$ require calculating eigendecompositons of $\mathbf{\Psi}_\ell$.  Thus, the eigendecomposition, and hence inverse, of $\mathbf{A}_\ell$ is readily available.  This answers Problem \ref{prob:estimability}.

Some methods, such as GmGM \cite{andrew_gmgm_2024}, require only a single eigendecomposition, and then find the optimum while staying in `eigen-space'.  They use the eigenvectors of covariance matrices; when the mean is updated, the new covariance matrices can be expressed in terms of rank-one-updates of the original covariance matrix.

Rank-one updates of eigendecompositions are cheap (comparatively); thus, for each flip-flop of our algorithm, if GmGM is used as the solver for $\mathbf{\Psi}_\ell$, then it can stay entirely in eigenspace - we preserve the `only-one-eigendecomposition' property of GmGM.  Other methods, such as TeraLasso, will still require multiple eigendecompositions, but these will remain useful for our calculation of $\mathbf{A}^{-1}_\ell$.  It is because of these convenient computational properties that we have chosen a flip-flop route for estimation of $\tilmu_\ell$, despite the fact that the optimization problem is jointly convex in $\{\tilmu_\ell\}_\ell$.

\begin{algorithm}
    \begin{algorithmic}
        \REQUIRE Dataset $\mathcal{D}$ of size $d_1, ..., d_K$
        \REQUIRE Sub-procedure $\bigoplus\mathrm{GRAPH}$ that estimates $\mathbf{\Psi}_\ell$
        \STATE \hspace{10pt} such as TeraLasso, GmGM\par
        \hspace{-10pt}\textbf{Output: } Identifiable MLE $\left(m, \tilmu_\ell, \tau, \widetilde{\mathbf{\Psi}}_\ell\right)$
        \STATE \textbf{\#Initialize Parameters with Reasonable Guesses}
        \STATE $m^\algiterate{0} \gets \frac{1}{d_\forall}\sum_{i=1}^{d_\forall} \mathrm{vec}\left[\mathcal{D}\right]_i$
        \FOR{$\ell \in \{1, ..., K\}$}
        \STATE Let $\mathbf{x}_\ell$ be vectorized $\mathcal{D}$ with $\ell$ as first axis.
        \STATE Initial guess is the mean along axis $\ell$:
        \STATE \hspace{10pt}$\tilmu^\algiterate{0}_\ell \gets \frac{1}{d_{\backslash\ell}} \sum_{i=1}^{d_{\backslash\ell}}\mathbf{x}_{\ell, i} - m^\algiterate{0}$
        \ENDFOR
        \STATE \textbf{\#Iterate Until Convergence}
        \WHILE{not converged}
        \STATE $\left\{\mathbf{\Psi}_\ell^\algiterate{i+1}\right\}_\ell \gets \bigoplus\mathrm{GRAPH}\left(\mathcal{D} - m^\algiterate{i}\mathbf{1} - \bigotilde_{\ell'}\boldsymbol{\mu}_\ell^\algiterate{i}\right)$
        \WHILE{not converged}
        \STATE \textbf{\#Find Mean Parameters}
        \STATE $m^\algiterate{i+1} \gets$ (Equation \ref{eq:m})
        \FOR{$\ell\in\left\{1,...,K\right\}$}
        \STATE $\tilmu_\ell^\algiterate{i + 1} \gets$ (Equation \ref{eq:tilmu})
        \ENDFOR
        \ENDWHILE
        \ENDWHILE
        \STATE Map $\left\{\mathbf{\Psi}_\ell\right\}$ to identifiable parameterization $\left(\left\{\widetilde{\mathbf{\Psi}}_\ell\right\}, \tau\right)$
    \end{algorithmic}
    \caption{\textbf{Noncentral KS-structured Gaussian MLE}}
\label{alg:mean-estim}
\end{algorithm}

See Algorithm \ref{alg:mean-estim} for a pseudocode presentation of our algorithm.  Note that the pseudocode does not take into account potential speedups that come from sharing information between the precision matrix and mean estimation tasks, such as sharing eigenvectors.  We wanted our method to be able to be used as a `drop-in wrapper' for pre-existing methods.  Accordingly, we have presented and implemented our methodology in its most general form.  For a graphical overview of how we optimize, see Figure \ref{fig:estimation-walkthrough}.

\section{Global Optimality}

The optimization procedure we consider is not convex, merely biconvex.  How do we know if the solution we find is globally optimal?  In this section, we will show that, despite the non-convexity, the problem still has the property that all local minima are global minima.  In fact, when identifiable representations are used, we will see that there is exactly one minimum.  As in the last section, we will defer the more mechanical details of the proof the the appendix, and rather give a sketch of the main ideas here.

\begin{theorem}
    \label{thm:unimodality}
    The maximum likelihood of the noncentral Kronecker-sum-structured normal distribution has a unique maximum, which is global.  The estimator defined by Algorithm \ref{alg:mean-estim} converges to this.
\end{theorem}
\begin{proof}
    See appendix.  Sketch given below.
\end{proof}

The general idea is to perform a one-to-one transformation on the problem to turn it into a strictly convex problem in a different set of parameters - specifically, an instance of conic programming with linear constraints.  Letting $\mathfrak{A}$ be our original problem, $\mathfrak{A}'$ be a slight modification of it, and $\mathfrak{B}$ be the conic problem, we show the following chain of implications:

\begin{align*}
    &\left(\boldsymbol{\omega}, \mathbf{\Omega}\right) \in \localmin{\mathfrak{A}} \\
    \iff& \left(\boldsymbol{\omega}, \mathbf{\Omega}, s=1\right) \in \localmin{\mathfrak{A}'} \\
    \iff& \left(\mathbf{\Gamma}\right) \in \localmin{\mathfrak{B}} \\
    \iff& \left(\mathbf{\Gamma}\right) \in \globmin{\mathfrak{B}} \tag{convexity}
\end{align*}


$\mathbf{\Gamma}$ is a one-to-one function of $\boldsymbol{\omega}, \mathbf{\Omega}$, and $s$, where $s$ is somewhat like a slack variable.  We show that local minima always occur at $s = 1$, and furthermore that at local minima, the value of the objective function of $\mathfrak{A}$ equals the value of the objective function of $\mathfrak{B}$ for the corresponding parameters.  Thus, all local minima of $\mathfrak{A}$ must obtain the same value, i.e. they are global minima.  In fact, because the transformation is one-to-one, and $\mathfrak{B}$ has a unique global minimum,  there is exactly one local minimum of $\mathfrak{A}$.

The transformation we use is:

\begin{align*}
    \boldsymbol{x} \sim \mathcal{N}\left(\boldsymbol{\omega}, \mathbf{\Omega}^{-1}\right) &\rightarrow \begin{bmatrix}\boldsymbol{x} \\ 1\end{bmatrix} \sim \mathcal{N}\left(\mathbf{0}, \mathbf{\Gamma}^{-1}\right) \\
    \mathbf{\Gamma} &= \begin{bmatrix}
        \mathbf{\Omega} & -\mathbf{\Omega}\boldsymbol{\omega} \\
        -\boldsymbol{\omega}^T\mathbf{\Omega} & \frac{1}{s} + \boldsymbol{\omega}^T\mathbf{\Omega}\boldsymbol{\omega}
    \end{bmatrix} \\
    s &> 0
\end{align*}

We first came across a similar transformation in a paper by \citeauthor{hosseini_matrix_2015} (\citeyear{hosseini_matrix_2015}), although their transformation was for the covariance matrix; this transform here can be derived from theirs using block matrix inversion.  The variable $\frac{1}{s}$ is used to allow $\mathbf{\Gamma}$ to be any positive definite matrix (subject to constraints on $\mathbf{\Omega}, \boldsymbol{\omega}$).  Otherwise, the value of the bottom right entry would be determined by the other entries.  It should be clear that, as long as $\mathbf{\Omega}$ is positive definite, this transformation is one-to-one given a fixed $s$.  If it is merely positive semidefinite, we cannot guarantee the global minimum is unique (but all local minima remain global); typically, solvers for Kronecker-sum-structured $\mathbf{\Omega}$ require it to be positive definite.

The strategy used to prove the chain of implications is to show that, when $s=1$, $\mathfrak{A}$ and $\mathfrak{A}'$ correspond to the same optimization problem.  Since $\left(\mathbf{\Omega}, \boldsymbol{\omega}, s\right) \rightarrow \mathbf{\Gamma}$ is a one-to-one mapping, as well as the fact that this mapping preserves the objective function's values, and finally the fact that $\mathbb{K}_{\mathbf{M}\mathbf{v}}$ is a linear subspace and hence preserves strict convexity of $\mathfrak{B}$, we have that there is exactly one minimum to $\mathfrak{B}$ which by the chain of implications is also the unique minimum of $\mathfrak{A}$.

All such claims are proven in the appendix.  In particular, uniqueness of the global minimum holds when one uses convex regularization penalties on $\mathbf{\Omega}$, and furthermore holds (subject to a reasonable constrant on $\boldsymbol{\omega}$) when one restricts $\mathbf{\Omega}$ to be a low-rank matrix, such as that considered by a variant of GmGM \cite{andrew_making_2024}.  The proof in the latter case is considerably more complicated, requiring mathematical machinery particularly from work by \citeauthor{hartwig_singular_1976} (\citeyear{hartwig_singular_1976}) and \citeauthor{holbrook_differentiating_2018} (\citeyear{holbrook_differentiating_2018}); it holds when we restrict $\boldsymbol{\omega}$ to be orthogonal to the nullspace of $\boldsymbol{\Omega}$ (a reasonable assumption since data generated from the distribution should also be orthogonal to the nullspace).

We claimed that our block coordinate descent algorithm converges to this unique optimal value.  Block coordinate descent methods are not guaranteed to converge for smooth nonconvex problems.  However, when there is a unique minimum for each block, as is the case here, block coordinate descent is known to converge \cite{tseng_convergence_2001, luenberger_linear_2008}.

\section{Results}
\label{sec:results}

\begin{table}[h!]
\centering
\begin{tabular}{ll}
\toprule
Figure & Mean Distribution \\
\midrule
Figure \ref{fig:ba-total-offset} & Constant mean of 1 \\
Figure \ref{fig:ba-structured} &  $m + \bigotilde_\ell \boldsymbol{\mu}_\ell$ \\
& $\boldsymbol{\mu}_\ell \sim \mathcal{N}(\mathbf{0}, \mathbf{I}),\text{ }m \sim \mathcal{N}(0, 1)$ \\
Figure \ref{fig:ba-unstructured} (left) & $\mathcal{N}(\mathbf{0}, \frac{1}{20}\mathbf{I})$ \\
Figure \ref{fig:ba-unstructured} (right) & $\frac{\boldsymbol{\omega}}{14}$ \\
& $\boldsymbol{\omega}_{ij} \sim \mathrm{Poisson}(10)$ \\
\bottomrule
\end{tabular}
\caption{The distributions of the means in each of our synthetic data experiments, and the corresponding figures.  Both mean distributions for Figure \ref{fig:ba-unstructured} have a variance of 0.05.  This variance had to be smaller for the unstructured distributions, otherwise the noise would overwhelm the signal in all algorithms.  We make no structural assumptions for Figure \ref{fig:ba-unstructured}; the mean is free to take any form, i.e. it is not KS-decomposable.}
\label{tab:synthetic-means}
\end{table}

All experiments were run on a 2020 MacBook Pro with an M1 chip and 8 GB of RAM.  Our method was implemented using NumPy 1.25.2 \cite{harris_array_2020} and SciPy 1.12.0 \cite{virtanen_scipy_2020}; for precision matrix routines, we used \citeauthor{greenewald_tensor_2019}'s reference implementation of TeraLasso (\citeyear{greenewald_tensor_2019}) as well as GmGM 0.5.3.  GmGM and TeraLasso are fairly similar algorithms; GmGM outputs positive definite matrices\footnote{Semi-definite in \citeauthor{andrew_making_2024} (\citeyear{andrew_making_2024}).} and does not use regularization, while TeraLasso uses regularization and does not require the outputs to be positive definite (as long as the Kronecker sum of the outputs is).  By presenting results with each algorithm, we hope to show that our proposed mean estimation algorithm is not specific to any one method of precision estimation.  All code is provided in the supplementary material.

In this section, we compare the performance of precision matrix algorithms with and without our mean estimation wrapper across a variety of test suites.  We first compare performances on synthetic data, in which the graphs are generated from a Barabasi-Albert (power-law) distribution.  We generate means for our synthetic data from a variety of distributions; see Table \ref{tab:synthetic-means}.  We also compare them on Erdos-Renyi graphs, achieving similar results, but for conciseness we defer those results to the appendix.  Next, we compare performances on the real-world COIL-20 video dataset (\citeauthor{nene_columbia_nodate}, 1996), demonstrating clear improvements.  Finally, we show that, without properly taking into account the mean, precision matrix estimation will yield clearly wrong results on transcriptomics datasets such as the mouse embryo stem cell `E-MTAB-2805' dataset \cite{buettner_computational_2015}.

\begin{figure}
    \centering
    \includegraphics[width=0.5\linewidth]{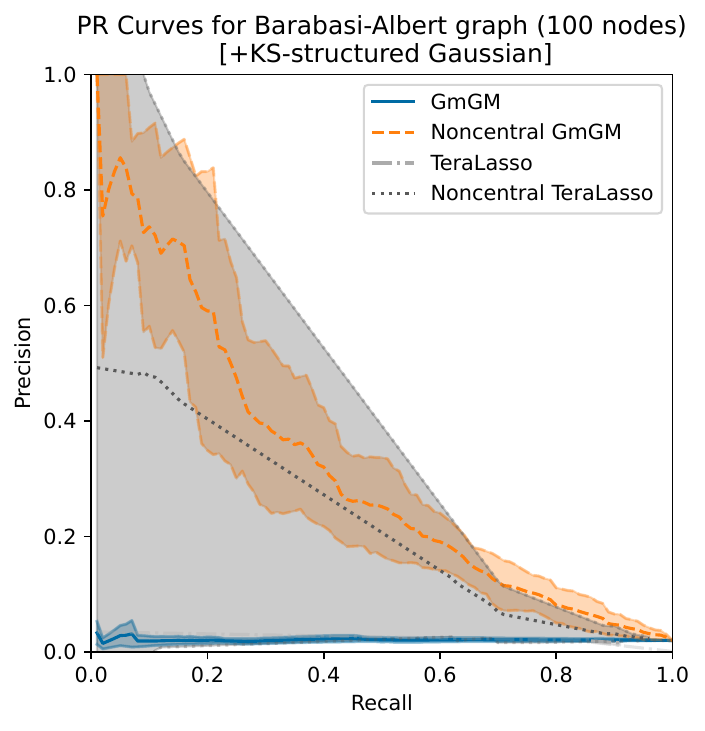}
    \caption{Precision and recall for on a synthetic dataset in which data was generated from a noncentral Kronecker-sum normal distribution.  Error bars are the best/worst performance over 10 trials; the center line is average performance.  Noncentral TeraLasso had large variance in its performance due to TeraLasso's high susceptibility to changes in the regularization parameter (changes as small as $10^{-8}$ can have large effects on the number of edges kept).  This led to our grid search being too coarse-grained during some trials.}
    \label{fig:ba-structured}
\end{figure}

For synthetic data, we already saw the dramatic deterioration in performance of algorithms that make the zero-mean assumption in Figure \ref{fig:ba-total-offset}, while mean-corrected algorithms still perform well.  That alone should produce cause for worry.  Zero-mean algorithms also perform poorly when we generate data from the noncentral Kronecker-sum-structured Gaussian distribution; understandably, our correction is immune to this problem (Figure \ref{fig:ba-structured}).

\begin{figure}[h!]
    \centering
    \begin{subfigure}[t]{0.45\linewidth}
        \centering
        \includegraphics[width=\linewidth]{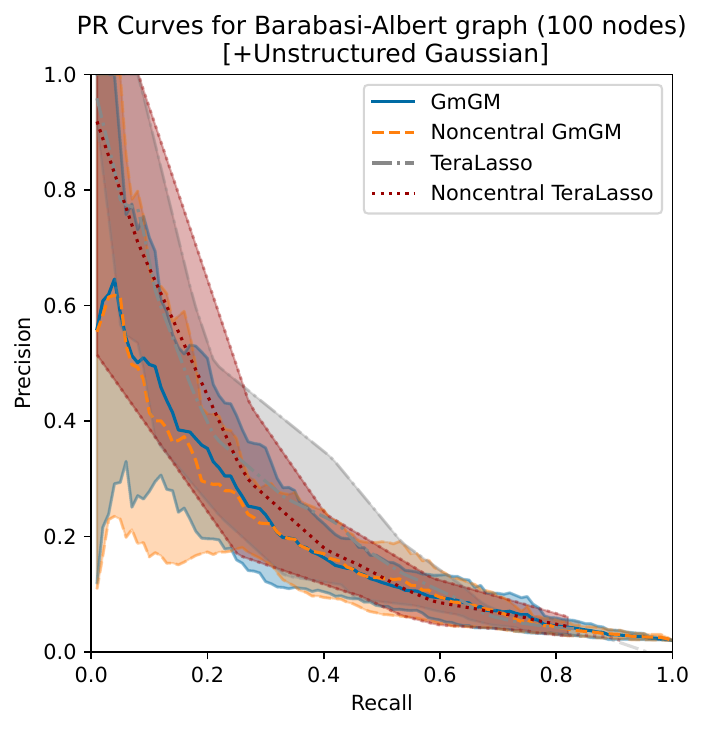}
    \end{subfigure}
    \begin{subfigure}[t]{0.45\linewidth}
        \centering
        \includegraphics[width=\linewidth]{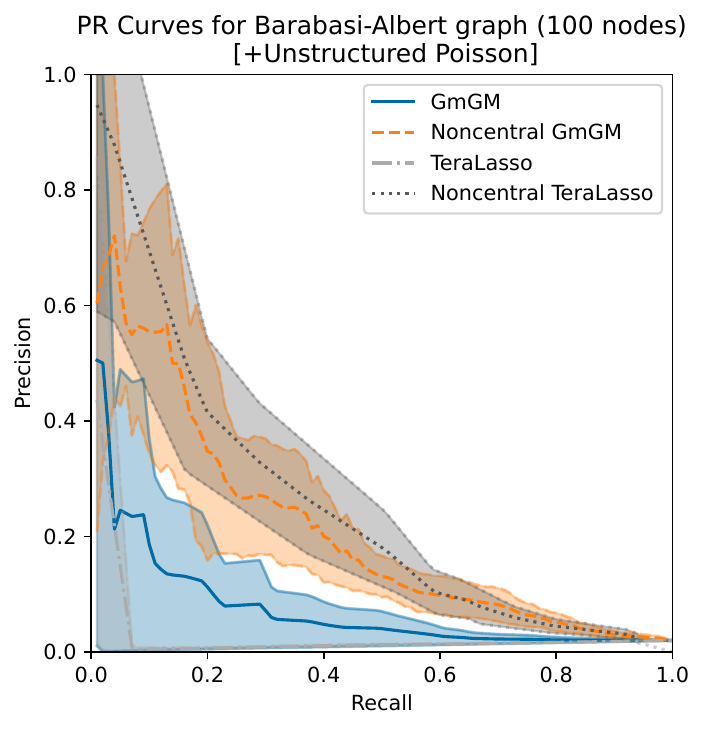}
    \end{subfigure}
    \caption{Precision and recall for on a synthetic dataset in which data was generated from a zero-mean Kronecker-sum normal distribution, with independent Gaussian (left) and Poisson (right) noise added to every element.  Error bars are the best/worst performance over 10 trials; the center line is average performance.}
    \label{fig:ba-unstructured}
\end{figure}

Finally, we may wonder what happens if our mean has absolutely no structure, with every element of our generated matrix having a different, independent mean.  We compared both Gaussian and Poisson-distributed random means (Figure \ref{fig:ba-unstructured}) and found that there was still a drop-off in performance for Poisson means, but not nearly as dramatic as that seen in the other cases considered.  There was no drop-off in the Gaussian case.

When we increased the variance of our means in our unstructured tests, both the standard algorithms and our mean-corrected algorithms dropped in performance equally.  We suspect this is because, for large variances, the signal is overwhelmed by the noise.  For small variances, since each row/column of the mean is uncorrelated, their effects average to zero, and thus our mean-corrected algorithms will also estimate means of zero (making them idendical to the uncorrected case).  The Poisson distribution behaves slightly differently, as the mean is instead a constant nonzero number, and hence our mean-corrected algorithm is able to correctly account for this, whereas the standard algorithms experience degradation in performance similar to that of the constant-mean offset test (Figure \ref{fig:ba-total-offset}), albeit to a lesser degree as the offset is smaller.

The unstructured mean experiment has important consequences: if real-world data has an unstructured mean, then it is less important to correct for the means!  However, it is likely that the mean does have a structured component, as each axis will have its own latent factors affecting the outcome in addition to any unstructured noise that may or may not exist.  To demonstrate this, we run our algorithm on two real-world datasets, showing that our mean correction truly does improve performance in practice.

\begin{figure}[h!]
    \centering
    \includegraphics[width=1\linewidth]{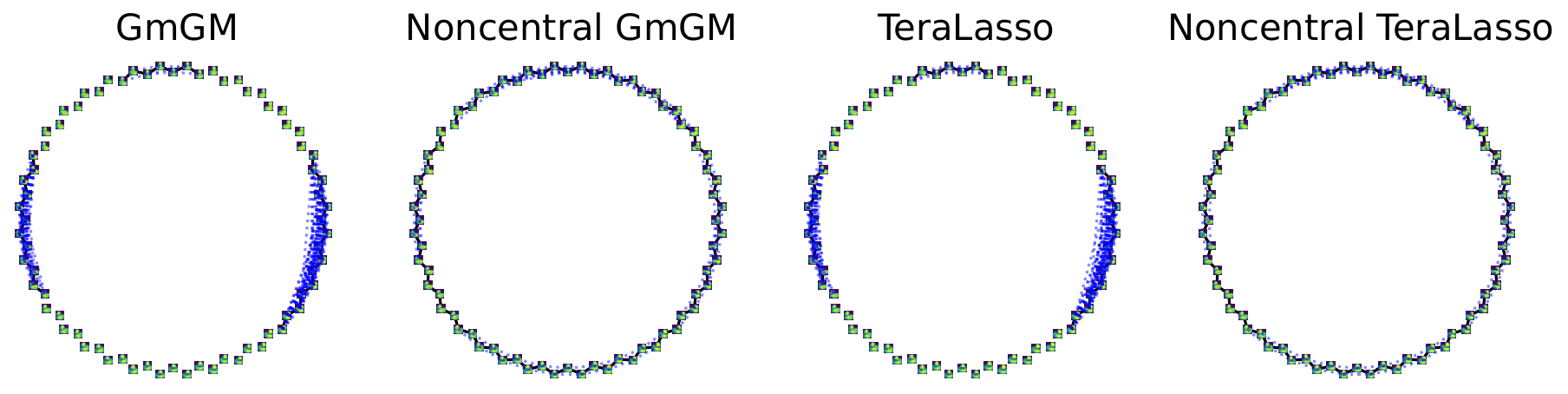}
    \caption{A comparison of the results on the COIL-20 dataset's frames.  Each circle represents the performance of a different algorithm; the prefix `noncentral' denotes that the algorithm uses our mean estimation method.  The images correspond to frames of the video (a duck rotating 360 degrees), arranged in order.  Dashed blue lines indicate a connection between non-adjacent frames; solid black lines indicate adjacent connections.}
    \label{fig:coil-results}
\end{figure}

The COIL-20 dataset consists of 20 videos (each with 72 frames) of objects rotating 360 degrees; it has become somewhat customary for multi-axis methods to present their performance on the rubber duck object of this dataset.  The dataset is chosen because of its ease to validate: frames that are close in time (i.e. adjacent) should be connected in the graph.  We present our results graphically in Figure \ref{fig:coil-results}, from which we can easily see the improvement.  We pre-processed the data using the nonparanormal skeptic \cite{liu_nonparanormal_2012} to relax the Gaussian assumption of the models, and we chose thresholding/regularization parameters such that there would be approximately 144 edges.

We considered `correct' edges to be those connecting adjacent frames or frames with a single frame in between.  GmGM had 59/147 (40\%) correct and TeraLasso had 55/142 (39\%) correct.  Once we wrap the methods in our mean estimation procedure, this rose to 127/147 (86\%) and 126/141 (89\%), respectively.  Furthermore, all `wrong' connections of our noncentral algorithms were nearly correct: with one exception, every connection was between frames with at most two frames between them.  For the standard algorithms, more than a third of all edges failed this criteria.  The noncentral algorithms also had high recall, at 88\% for both noncentral GmGM and TeraLasso.  The uncorrected algorithms had a recall of 41\% and 38\%.  Overall, the results for this experiment are fairly conclusively in favor of using our correction.

Our final experiment is on the E-MTAB-2805 scRNA-seq dataset.  This dataset consists of 288 mouse embryo stem cells and 34,573 unique genes, with each cell being labeled by its stage in the cell cycle (G1, S, or G2M).  It has previously been considered by the creators of scBiGLasso \cite{li_scalable_2022}, in which they limited it to a hand-selected set of 167 mitosis-related genes - we limited it to the same set of genes.  We would expect that cells at the same stage in the same cycle should have some similarities, and hence should have some tendency to cluster together in our learned graphs.  This tendency may be weak - scRNA-seq data is very noisy, and many other biological signals are competing for influence on a cell's expression profile - but it should exist.

We can use `assortativity' as a measure for the tendency for cells within a stage to connect: assortativity ranges from -1 to 1, and corresponds to a Pearson correlation coefficient of `within-stage' degrees of adjacent nodes.  When it is positive, it indicates the strength of the tendency of nodes in the same cell cycle stage to connect to each other.  If it is negative, cells in differing cell cycle stages tend to be connected to each other.  If zero, there is no tendency; the connections are effectively random.

\begin{figure}
    \centering
    \begin{subfigure}[t]{0.45\linewidth}
        \centering
        \includegraphics[width=\linewidth]{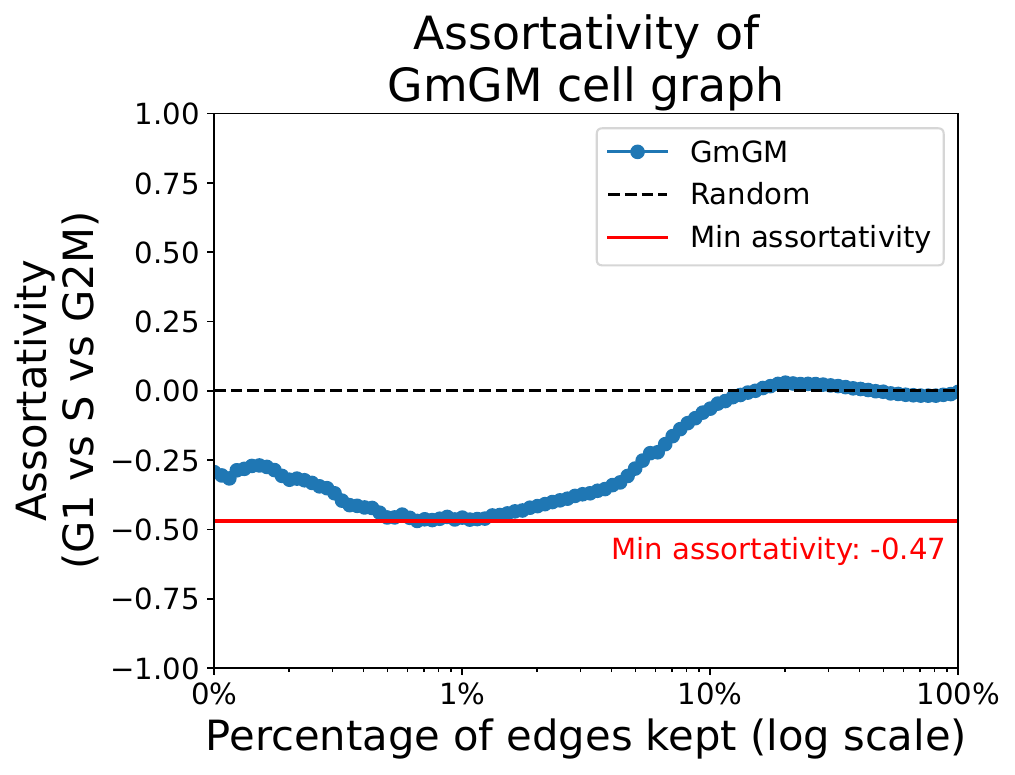}
    \end{subfigure}
    \begin{subfigure}[t]{0.45\linewidth}
        \centering
        \includegraphics[width=\linewidth]{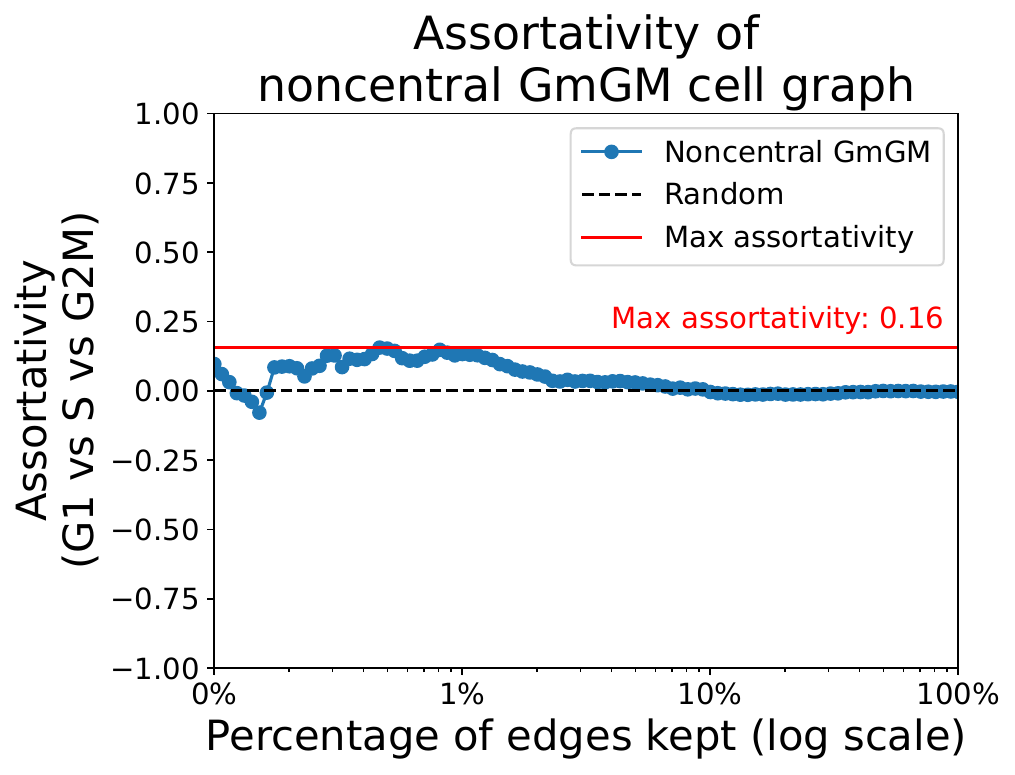}
    \end{subfigure}
    \caption{Assortativity of standard GmGM and our mean-corrected GmGM as we vary the number of edges kept, on the E-MTAB-2805 cell-cycle dataset.  The cells were labeled by their stage in the cell cycle (G1/S/G2M).}
    \label{fig:e-mtab-2805}
\end{figure}

When using standard algorithms, without correcting for the mean, we find that, not only is there not a tendency for cells to connect to same-stage cells, but there is a very strong tendency for them \textit{not to connect} (assortativity = -0.47).  This runs contrary to intuition: cells with some factor in common should not be caused by that same factor to look different.  When correcting for the means, we instead see a modest positive tendency (assortativity = 0.16), which is more reasonable.  Figure \ref{fig:e-mtab-2805} shows the change in assortativity for GmGM, with and without our correction, as we vary the threshold of how many edges to keep.  This shows that our findings are not dependent on our choice of penalty, but rather are a consistent feature of the problem.  We did not run the experiment with TeraLasso due to its slow runtime, but the results for it would likely be broadly similar (as they have been in the rest of this paper).

\section{Conclusion}

Current Kronecker-structured models do not take into account the mean of the data.  As one might expect, this can lead to wildly incorrect inferences, as best exemplified in Figures \ref{fig:ba-total-offset} and \ref{fig:e-mtab-2805}.  As shown in these examples, failure to take into account the mean not only results in incorrect inferences, but in fact it can result in inferences directly opposite to reality: as thresholding/regularization increased, algorithms that did not take into account the mean decreased in both precision and recall, and identified an extremely strong (and incorrect) repulsive force between cells in the same stage of the cell cycle.

In this paper, we have demonstrated that this can be fixed by adding a mean-estimating wrapper around standard algorithms.  While the means and precisions are nonseperable in the KS-structured case, and hence coincide with a nonconvex optimization problem, we proved that there is a unique global optimum, to which our algorithm converges.  We implemented our wrapper as a `drop-in' wrapper (so that it can easily be used with any `black box' standard precision-estimating algorithm), but we also point to aspects of our method that make it easier to more tightly integrate it into precision-estimating frameworks.  In particular, it can be made to preserve the `only-one-eigendecomposition' property of GmGM, which is highly beneficial for scalability.

As future work, we aim to use this model to better understand gene regulatory networks and the cellular microenvironment in single-cell omics data.  We would also like to generalize the method to the multi-modal case, as in GmGM.

\bibliography{references}

\appendix

\section{Deferred Proofs}
\label{apn:proofs}

In this portion of the appendix, we will formally prove the claims presented in the paper, namely we will prove that the locations of coordinate-wise minima reported in the paper are indeed accurate, and furthermore that the full optimization problem has a unique solution.

\subsection{Derivation of Gradients}

First, we will derive gradients with respect to all parameters.  Recall that the negative log-likelihood is, up to an additive constant:

\begin{align*}
    \mathrm{NLL}\left(\mathcal{D}\right) &\propto \frac{-1}{2}\log\left|\mathbf{\Omega}\right| + \frac{1}{2}\left(\boldsymbol{x} - \boldsymbol{\omega}\right)^T\mathbf{\Omega}\left(\boldsymbol{x} - \boldsymbol{\omega}\right)
\end{align*}

We do not need to work out the gradient w.r.t. $\mathbf{\Psi}_\ell$, as we defer that optimization task to algorithms such as TeraLasso.  Thus, we are left with the following quadratic function:

\begin{align*}
    \mathrm{NLL}\left(\mathcal{D}\right) &\propto \frac{1}{2}\left(\boldsymbol{x} - \boldsymbol{\omega}\right)^T\mathbf{\Omega}\left(\boldsymbol{x} - \boldsymbol{\omega}\right)
\end{align*}

We will differentiate with respect to the identifiable parameters $\tilmu_\ell, m$.  Generically, for any parameter $p$, the derivative will look like:

\begin{align*}
    \pderiv{p}\mathrm{NLL}\left(\mathcal{D}\right) &= -\left(\boldsymbol{x} - \boldsymbol{\omega}\right)^T\mathbf{\Omega}\left(\pderiv{p}\boldsymbol{\omega}\right)
\end{align*}

Oberve that:

\begin{align*}
    \pderiv{m}\boldsymbol{\omega} &= \pderiv{m} \left(m\mathbf{1}_{d_\forall} + \bigotilde_\ell \tilmu_\ell\right) \\
    &= \mathbf{1}_{d_\forall} \\
    \pderiv{\tilmu_\ell}\boldsymbol{\omega} &= \pderiv{\tilmu_\ell} \bigotilde_\ell \tilmu_\ell \\
    &= \pderiv{\tilmu_\ell} \left(\tilmu_\ell \otilde \tilmu_{\backslash\ell}\right) \tag{up to permutation} \\
    &= \pderiv{\tilmu_\ell} \left(\tilmu_\ell \otimes \mathbf{1}_{d_{\backslash\ell}} + \mathbf{1}_{d_\ell} \otimes \tilmu_{\backslash\ell}\right) \\
    &= \mathbf{I}_{d_\ell} \otimes \mathbf{1}_{d_{\backslash\ell}}
\end{align*}

The latter derivation is also obvious from our ability to rewrite $\tilmu_\ell \otimes \mathbf{1}_{d_{\backslash\ell}}$ as $\left(\mathbf{I}_{d_\ell} \otimes \mathbf{1}_{d_{\backslash\ell}}\right)\tilmu_\ell$.  We gave the derivative up to permutation, for notational simplicity (since we are only interested in when it equals zero).  Without the permutation, it is easy to see that it is $\mathbf{1}_{d_{<\ell}} \otimes \mathbf{I}_{d_\ell} \otimes \mathbf{1}_{d_{>\ell}}$.  Finally, note that the latter derivative would be analogous if done w.r.t. the unconstrained parameters $\boldsymbol{\mu}_\ell$.

\subsection{Derivation of the Coordinate-wise Minimum of $m$}

To derive the coordinate-wise optimum of $m$, we will find when the gradient is zero.

\begin{align*}
    \pderiv{m}\mathrm{NLL}\left(\mathcal{D}\right) &= -\left(\boldsymbol{x} - \boldsymbol{\omega}\right)^T\mathbf{\Omega}\mathbf{1}_{d_\forall} \\
    &= -\left(\boldsymbol{x} - \boldsymbol{\omega}\right)^T\left(\bigoplus_\ell \mathbf{\Psi}_\ell\right)\left(\bigotimes_\ell\mathbf{1}_{d_\ell}\right) \\
    &= -\left(\boldsymbol{x} - \boldsymbol{\omega}\right)^T\left(\bigotilde_\ell \mathbf{\Psi}_\ell\mathbf{1}_\ell\right) \\
    &= \left(m\mathbf{1}_{d_\forall} - \boldsymbol{x} + \bigotilde_\ell\tilmu_\ell\right)^T\left(\bigotilde_\ell \mathbf{\Psi}_\ell\mathbf{1}_\ell\right) \\
\end{align*}

Before proceeding, let's prove a helpful lemma:

\begin{lemma}
    \label{lem:sum-of-kronsum}
    For any set of $d_\ell \times d_\ell$ matrices $\mathbf{M}_\ell$:
    
    \begin{align*}
        \mathbf{1}_{d_\forall}^T\left(\bigoplus_\ell\mathbf{M}_\ell\right)\mathbf{1}_{d_\forall} = \sum_\ell d_{\backslash\ell}\mathbf{1}_{d_\ell}^T\mathbf{M}_\ell\mathbf{1}_{d_\ell}
    \end{align*}
\end{lemma}
\begin{proof}
    \begin{align*}
        \mathbf{1}_{d_\forall}^T\left(\bigoplus_\ell\mathbf{M}_\ell\right)\mathbf{1}_{d_\forall} &= \mathbf{1}_{d_\forall}^T\left(\sum_\ell\mathbf{I}_{d_{<\ell}} \otimes \mathbf{M}_\ell\otimes \mathbf{I}_{d_{>\ell}}\right)\mathbf{1}_{d_\forall}
    \end{align*}
    \begin{align*}
        &= \sum_\ell \mathbf{1}_{d_{<\ell}}^T\mathbf{1}_{d_{<\ell}} \otimes \mathbf{1}_{d_\ell}^T\mathbf{M}_\ell\mathbf{1}_{d_\ell} \otimes \mathbf{1}_{d_{>\ell}}^T\mathbf{1}_{d_{>\ell}} \\
        &= \sum_\ell d_{<\ell}d_{>\ell}\mathbf{1}_{d_\ell}^T\mathbf{M}_\ell\mathbf{1}_{d_\ell} \\
        &= \sum_\ell d_{\backslash\ell}\mathbf{1}_{d_\ell}^T\mathbf{M}_\ell\mathbf{1}_{d_\ell}
    \end{align*}
\end{proof}

Recall that $\mathbf{1}_{d_\ell}^T\mathbf{\Psi}_\ell\mathbf{1}_{d_\ell} \overset{\mathrm{def}}{=} \theta_\ell$.  At zero, the equality can be reshaped into the following:

\begin{align*}
    m\left(\bigotimes_\ell \mathbf{1}_{d_\ell}^T\right)\bigotilde_\ell \mathbf{\Psi}_\ell\mathbf{1}_\ell &= \left(\mathbf{x} - \bigotilde_\ell\tilmu_\ell\right)^T\bigotilde_\ell \mathbf{\Psi}_\ell\mathbf{1}_\ell \\
    \implies m\sum_\ell d_{\backslash\ell}\theta_\ell &= \left(\mathbf{x} - \bigotilde_\ell\tilmu_\ell\right)^T\bigotilde_\ell \mathbf{\Psi}_\ell\mathbf{1}_\ell \tag{Lemma \ref{lem:sum-of-kronsum}}\\
    \implies m &= \frac{\left(\mathbf{x} - \bigotilde_\ell\tilmu_\ell\right)^T\bigotilde_\ell \mathbf{\Psi}_\ell\mathbf{1}_\ell}{\sum_\ell d_{\backslash\ell}\theta_\ell}
\end{align*}

\subsection{Derivation of the Coordinate-wise Minima of $\tilmu_\ell$}

Rather than considering the gradients, we will frame this problem as an equality-constrained QP problem, which has closed-form minima.  Here, we will let $\mathbf{x}_\ell$ be our vectorized dataset, vectorized in an order such that axis $\ell$ was considered the first axis.  We'll further write $\widetilde{\mathbf{x}}_\ell = \mathbf{x_\ell} - m\mathbf{1}_{d_\forall}$ to keep our equations compact, and represent $\bigotilde_{\ell'}\tilmu_{\ell'}$ as $\tilmu_\ell \otilde \tilmu_{\backslash_\ell}$.

\begin{align*}
    \underset{\tilmu_\ell}{\mathrm{argmin}}&\hspace{5pt} \frac{1}{2}\left(\widetilde{\mathbf{x}}_\ell - \bigotilde_{\ell'}\tilmu_{\ell'}\right)^T\mathbf{\Omega}\left(\widetilde{\mathbf{x}}_\ell - \bigotilde_{\ell'}\tilmu_{\ell'}\right) \\
    \mathrm{where}&\hspace{5pt} \tilmu_\ell^T\mathbf{1}_{d_\ell} = \mathbf{0}
\end{align*}

We then expand the equation, removing terms that do not involve $\tilmu_\ell$ and recalling that $\mathbf{u}^T\mathbf{M}\mathbf{v} = \mathbf{v}^T\mathbf{M}\mathbf{u}$.

\begin{align*}
    \underset{\tilmu_\ell}{\mathrm{argmin}}&\hspace{2pt} \frac{1}{2} \left(\bigotilde_{\ell'}\tilmu_{\ell'}^T\right)\mathbf{\Omega}\left(\bigotilde_{\ell'}\tilmu_{\ell'}\right)  - \widetilde{\mathbf{x}}_\ell^T \mathbf{\Omega}\left(\bigotilde_{\ell'}\tilmu_{\ell'}\right) \\
    \mathrm{where}&\hspace{5pt} \tilmu_\ell^T\mathbf{1}_{d_\ell} = \mathbf{0}
\end{align*}

Again, we split $\bigotilde_{\ell'}\tilmu_{\ell'}$ into $\tilmu_\ell \otilde \tilmu_{\backslash_\ell}$ and remove terms not containing $\tilmu_\ell$.

\begin{align*}
    \underset{\tilmu_\ell}{\mathrm{argmin}}&\hspace{2pt}\left\{\begin{matrix}
        \frac{1}{2}\left(\tilmu_\ell \otimes \mathbf{1}_{d_{\backslash\ell}}\right)^T\mathbf{\Omega}\left(\tilmu_\ell \otimes \mathbf{1}_{d_{\backslash\ell}}\right) \\
        + \left(\mathbf{1}_{d_\ell} \otimes \tilmu_{{\backslash\ell}} - \widetilde{\mathbf{x}}_\ell\right)^T\mathbf{\Omega}\left(\tilmu_\ell \otimes \mathbf{1}_{d_{\backslash\ell}}\right)
    \end{matrix}\right.\\
    \mathrm{where}&\hspace{5pt} \tilmu_\ell^T\mathbf{1}_{d_\ell} = \mathbf{0}
\end{align*}

Finally, recall that $\tilmu_\ell \otimes \mathbf{1}_{d_{\backslash\ell}} = \left(\mathbf{I}_{d_\ell}\otimes\mathbf{1}_{\backslash\ell}\right)\tilmu_\ell$.  Let's define the following constants:

\begin{align}
    \mathbf{A}_\ell &= \left(\mathbf{I}_{d_\ell}\otimes\mathbf{1}_{d_{\backslash\ell}}\right)^T\mathbf{\Omega}\left(\mathbf{I}_{d_\ell}\otimes\mathbf{1}_{d_{\backslash\ell}}\right) \label{eq:a-def}\\
    \mathbf{b}_\ell &= \left(\mathbf{1}_{d_\ell} \otimes \tilmu_{{\backslash\ell}} - \widetilde{\mathbf{x}}_\ell\right)^T\mathbf{\Omega}\left(\mathbf{I}_{d_\ell}\otimes\mathbf{1}_{d_{\backslash\ell}}\right) \label{eq:b-def}
\end{align}

Note that this definition is different from the one in the main paper; in the next section we will show that they are equivalent.  With these constants, we can rewrite our minimization problem as:

\begin{align*}
    \underset{\tilmu_\ell}{\mathrm{argmin}}&\hspace{5pt}
        \frac{1}{2}\tilmu_\ell^T \mathbf{A}_\ell \tilmu_\ell + \mathbf{b}_\ell^T\tilmu_\ell\\
    \mathrm{where}&\hspace{5pt} \tilmu_\ell^T\mathbf{1}_{d_\ell} = \mathbf{0}
\end{align*}

This is an equality-constrained QP problem; these have closed-form solutions.  Letting $\lambda$ be the Lagrange multiplier that appears when considering the constraint, we know that the solution is the solution to a linear equation:

\begin{align*}
    \begin{bmatrix}
        \mathbf{A}_\ell & \mathbf{1}_{d_\ell} \\
        \mathbf{1}_{d_\ell}^T & \mathbf{0}
    \end{bmatrix}\begin{bmatrix}
        \tilmu_\ell \\ \lambda
    \end{bmatrix} &= \begin{bmatrix}
        -\mathbf{b}_\ell \\ \mathbf{0}
    \end{bmatrix}
\end{align*}

We then arrive at the solution with some simple algebra.

\begin{align*}
    &&\mathbf{A}_\ell\tilmu_\ell + \mathbf{1}_{d_\ell}\lambda &= -\mathbf{b}_\ell \\
    &\implies&\mathbf{1}_{d_\ell}^T\tilmu_\ell + \mathbf{1}^T_{d_\ell}\mathbf{A}^{-1}_\ell\mathbf{1}_{d_\ell}\lambda &= -
    \mathbf{1}_{d_\ell}^T\mathbf{A}^{-1}_\ell\mathbf{b}_\ell \\
    &\implies& \lambda &= \frac{-
    \mathbf{1}_{d_\ell}^T\mathbf{A}^{-1}_\ell\mathbf{b}_\ell}{
    \mathbf{1}_{d_\ell}^T\mathbf{A}^{-1}_\ell\mathbf{1}_{d_\ell}} \\
    &\implies& \mathbf{A}_\ell\tilmu_\ell &= \frac{
    \mathbf{1}_{d_\ell}^T\mathbf{A}^{-1}_\ell\mathbf{b}_\ell}{
    \mathbf{1}_{d_\ell}^T\mathbf{A}^{-1}_\ell\mathbf{1}_{d_\ell}}\mathbf{1}_{d_\ell} - \mathbf{b}_\ell
\end{align*}

Which leads us to the following closed-form solution:

\begin{align*}
    \tilmu_\ell &= \frac{
    \mathbf{1}_{d_\ell}^T\mathbf{A}_\ell^{-1}\mathbf{b}_\ell}{
    \mathbf{1}_{d_\ell}^T\mathbf{A}_\ell^{-1}\mathbf{1}_{d_\ell}}\mathbf{A}_\ell^{-1}\mathbf{1}_{d_\ell} - \mathbf{A}_\ell^{-1}\mathbf{b}_\ell 
\end{align*}

\subsection{Derivation of $\mathbf{A}_\ell$ and $\mathbf{b}_\ell$}

While the derivation of $\mathbf{A}_\ell$ and $\mathbf{b}_\ell$ given in the previous section yields a formula for the coordinate-wise minima of $\tilmu_\ell$, calculating $\mathbf{A}_\ell$ and $\mathbf{b}_\ell$ using Equations \ref{eq:a-def} and \ref{eq:b-def} requires multiplying together very large matrices ($\mathbf{\Omega}$ is $d_\forall \times d_\forall$).  In this section, we will produce a simpler formula involving only $d_\ell \times d_\ell$ matrix manipulations.

We will make use of the fact that $\mathbf{\Omega}$ has a lower dimensional parameterization in terms of $\mathbf{\Psi}_\ell$.  We could instead have chosen to parameterize it in terms of the identifiable parameterization $\left(\widetilde{\mathbf{\Psi}}_\ell, \tau\right)$, but there is not much gain in doing so; it would merely make the formulas presented here more complicated.

\begin{align*}
    \mathbf{A}_\ell \hspace{-2pt}&= \left(\mathbf{I}_{d_\ell}\otimes\mathbf{1}_{d_{\backslash\ell}}\right)^T\mathbf{\Omega}\left(\mathbf{I}_{d_\ell}\otimes\mathbf{1}_{d_{\backslash\ell}}\right) \\
    &=\hspace{-2pt} \left(\mathbf{I}_{d_\ell}\otimes\mathbf{1}_{d_{\backslash\ell}}\right)^T\hspace{-4pt}\left(\mathbf{\Psi}_\ell \otimes \mathbf{I}_{d_{\backslash\ell}} + \mathbf{I}_{d_\ell} \otimes \mathbf{\Psi}_{{\backslash\ell}}\right)\hspace{-2pt}\left(\mathbf{I}_{d_\ell}\otimes\mathbf{1}_{d_{\backslash\ell}}\right) \\
    &= \left(\mathbf{\Psi}_\ell \otimes \mathbf{1}_{d_{\backslash\ell}}^T\mathbf{1}_{d_{\backslash\ell}} + \mathbf{I}_{d_\ell}\otimes\mathbf{1}_{d_{\backslash\ell}}^T\mathbf{\Psi}_{\backslash\ell}\mathbf{1}_{d_{\backslash\ell}}\right) \\
    &= d_{\backslash_\ell}\mathbf{\Psi}_\ell + \sum_{\ell'\neq\ell} \frac{d_\forall}{d_\ell d_{\ell'}}\theta_\ell\mathbf{I}_{d_\ell} \tag{Lemma \ref{lem:sum-of-kronsum}} \\
    &= d_{\backslash_\ell}\mathbf{\Psi}_\ell + \theta_{\backslash\ell}\mathbf{I}_{d_\ell}
\end{align*}

\begin{align*}
    \mathbf{b}_\ell &= \left(\mathbf{1}_{d_\ell} \otimes \tilmu_{{\backslash\ell}} - \widetilde{\mathbf{x}}_\ell\right)^T\mathbf{\Omega}\left(\mathbf{I}_{d_\ell}\otimes\mathbf{1}_{d_{\backslash\ell}}\right) \\
    &= \left(\mathbf{1}_{d_\ell} \otimes \tilmu_{{\backslash\ell}} - \widetilde{\mathbf{x}}_\ell\right)^T\left(\mathbf{\Psi}_{\ell}\otimes\mathbf{1}_{d_{\backslash\ell}} + \mathbf{I}_{d_\ell}\otimes\mathbf{\Psi}_{\backslash\ell}\mathbf{1}_{d_{\backslash\ell}}\right) \\
    &= \mathbf{1}_{d_\ell}^T\mathbf{\Psi}_{\ell}\otimes\tilmu_{\backslash\ell}^T\mathbf{1}_{d_{\backslash\ell}} + \mathbf{1}_{d_\ell}\otimes\tilmu_{d_{\backslash\ell}}^T\mathbf{\Psi}_{\backslash\ell}\mathbf{1}_{d_{\backslash\ell}} \\
    &\hspace{10pt}-\widetilde{\mathbf{x}}_\ell^T\left(\mathbf{\Psi}_{\ell}\otimes\mathbf{1}_{d_{\backslash\ell}} + \mathbf{I}_{d_\ell}\otimes\mathbf{\Psi}_{\backslash\ell}\mathbf{1}_{d_{\backslash\ell}}\right)
\end{align*}

Here, we can observe that $\tilmu_{\backslash\ell}^T\mathbf{1}_{d_{\backslash\ell}} = 0$ due to our constraints.  Thus, we get:

\begin{align*}
    \mathbf{b}_\ell &= \left[\mathbf{1}_{d_{\backslash\ell}}^T\mathbf{\Psi}_\ell\mathbf{1}_{d_{\backslash\ell}}\right]\mathbf{1}_{d_\ell}\hspace{-2pt}- \widetilde{\mathbf{x}}_\ell^T\left(\mathbf{\Psi}_{\ell}\otimes\mathbf{1}_{d_{\backslash\ell}}\hspace{-1pt}+ \mathbf{I}_{d_\ell}\otimes\mathbf{\Psi}_{\backslash\ell}\mathbf{1}_{d_{\backslash\ell}}\right) \\
\end{align*}

Expanding out $\widetilde{\mathbf{x}}_\ell$ into $\mathbf{x}_\ell - m\mathbf{1}_{d_\forall}$ yields the definition of $b_\ell$ given in the main paper.

\begin{align*}
    &m\mathbf{1}_{d_\forall}^T \left(\mathbf{\Psi}_{\ell}\otimes\mathbf{1}_{d_{\backslash\ell}}+ \mathbf{I}_{d_\ell}\otimes\mathbf{\Psi}_{\backslash\ell}\mathbf{1}_{d_{\backslash\ell}}\right) \\
    &= m d_{\backslash\ell}\mathbf{1}_{d_\ell}^T\mathbf{\Psi}_\ell + m\theta_{\backslash\ell}\mathbf{1}_{d_{\backslash\ell}}
\end{align*}

Unfortunately, this still involves a large multiplication; can we express it in a simpler way?

We can, but it involves introducing the concept of a matricization of a tensor: for a $d_1 \times ... \times d_K$ tensor $\mathcal{T}$, the $\ell$-matricization of it, $\mathrm{mat}_\ell\left[\mathcal{T}\right]$ is a $d_\ell \times d_{\backslash\ell}$ matrix constructed by stacking all axes together except $\ell$.  This allows us a more formal definition of $\mathbf{x}_\ell$: $\mathbf{x}_\ell \overset{\mathrm{def}}{=} \mathrm{vec}\hspace{3pt}\mathrm{mat}_\ell\left[\mathcal{D}\right]$.

We can then utilize a well-known convenient relation between Kronecker products and vectorization\footnote{There are two conventions for $\mathrm{vec}$, rows-first or columns-first; that convention affects the order of multiplication presented here.  We use the rows-first convention as that aligns with the convention used for $\mathrm{mat}$.}:

\begin{align*}
    \mathrm{vec}\left[\mathbf{M}\right]^T\left(\mathbf{U}\otimes \mathbf{V}\right) &= \mathrm{vec}\left[\mathbf{U}^T\mathbf{M}\mathbf{V}\right]^T
\end{align*}

Note that $\mathbf{\Psi}_\ell$ is symmetric, so the transpose is irrelevant.  This allows us to simplify.

\begin{align*}
    &\mathrm{vec}\hspace{3pt}\mathrm{mat}\left[\mathcal{D}\right]^T\left(\mathbf{\Psi}_{\ell}\otimes\mathbf{1}_{d_{\backslash\ell}}+ \mathbf{I}_{d_\ell}\otimes\mathbf{\Psi}_{\backslash\ell}\mathbf{1}_{d_{\backslash\ell}}\right) \\
    &= \mathrm{vec}\left[\mathbf{\Psi}_\ell\mathrm{mat}_\ell\left[\mathcal{D}\right]\mathbf{1}_{d_{\backslash\ell}} \right]^T + \mathrm{vec}\left[\mathrm{mat}_\ell\left[\mathcal{D}\right]\mathbf{\Psi}_{\backslash\ell}\mathbf{1}_{d_{\backslash\ell}} \right]^T
\end{align*}

As a final step for simplification, we recall another relation from \citeauthor{kolda_tensor_2009} (\citeyear{kolda_tensor_2009}).  First, let $\mathcal{D} \times_\ell \mathbf{M}$ represent multiplying a tensor $\mathcal{D}$ along its $\ell$th axis by a matrix $\mathbf{M}$.  Note that for matrices, $\mathbf{N} \times_0 \mathbf{M} = \mathbf{M}^T\mathbf{N}$ and $\mathbf{N} \times_1 \mathbf{M} = \mathbf{N}\mathbf{M}$.  Intuitively, it corresponds to batch matrix multiplication when the tensor has more than two axes.

\begin{align*}
    \mathrm{mat}_\ell\left[\mathcal{T}\right]\bigotimes_{\ell'\neq\ell}\mathbf{M}_{\ell'} = \mathrm{mat}_\ell\left[\mathcal{T} \bigtimes_{\ell'\neq\ell}\mathbf{M}_{\ell'}\right]
\end{align*}

This allows us to rewrite our expression as:

\begin{align*}
    &\hspace{5pt}\mathrm{vec}\left[\mathbf{\Psi}_\ell\mathrm{mat}_\ell\left[\mathcal{D}\right]\mathbf{1}_{d_{\backslash\ell}} \right]^T + \mathrm{vec}\left[\mathrm{mat}_\ell\left[\mathcal{D}\right]\mathbf{\Psi}_{\backslash\ell}\mathbf{1}_{d_{\backslash\ell}} \right]^T \\
    &= \mathrm{vec}\left[\mathbf{\Psi}_\ell\mathrm{mat}_\ell\left[\mathcal{D}\right]\mathbf{1}_{d_{\backslash\ell}} \right]^T
    \\&\hspace{5pt}+ \sum_{\ell'\neq\ell}\mathrm{vec}\left[\mathrm{mat}_\ell\left[\mathcal{D} \times_{\ell'}\mathbf{\Psi}_{\ell'}\right]\mathbf{1}_{d_{\backslash\ell}} \right]^T
\end{align*}

These always result in vectors, so we can drop the $\mathrm{vec}$ prefix.

\begin{align*}
    \mathbf{\Psi}_\ell\mathrm{mat}_\ell\left[\mathcal{D}\right]\mathbf{1}_{d_{\backslash\ell}}+ \sum_{\ell'\neq\ell}\mathrm{mat}_\ell\left[\mathcal{D} \times_{\ell'}\mathbf{\Psi}_{\ell'}\right]\mathbf{1}_{d_{\backslash\ell}}
\end{align*}

This looks complicated, but has a fairly straightforward interpretation.  The $\mathbf{\Psi}_\ell\mathrm{mat}_\ell\left[\mathcal{D}\right]\mathbf{1}_{d_{\backslash\ell}}$ term corresponds to summing over all axes of $\mathcal{D}$ except $\ell$, and then transforming the sum by $\mathbf{\Psi}_\ell$.  Each $\mathrm{mat}_\ell\left[\mathcal{D} \times_{\ell'}\mathbf{\Psi}_{\ell'}\right]\mathbf{1}_{d_{\backslash\ell}}$ term can be interpreted as a batch matrix multiplication along axis $\ell'$, before summing over all axes but $\ell$.  Of course, in practice it is more efficient to perform some sums first before the multiplication, as given below, but it causes rather cumbersome notation.

\begin{align*}
    \sum_{\ell'\neq\ell}\mathrm{mat}_\ell\left[\left(\mathcal{D} \bigtimes_{\ell\neq\ell''\neq\ell'} \mathbf{1}_{\backslash\ell''}\right) \times_{\ell'}\mathbf{\Psi}_{\ell'}\mathbf{1}_{\backslash\ell'}\right]
\end{align*}

To summarize, we have that:

\begin{align*}
    \mathbf{b}_\ell &= \left[\mathbf{1}_{d_{\backslash\ell}}^T\mathbf{\Psi}_\ell\mathbf{1}_{d_{\backslash\ell}}\right]\mathbf{1}_{d_\ell} + m d_{\backslash\ell}\mathbf{1}_{d_\ell}^T\mathbf{\Psi}_\ell + m\theta_{\backslash\ell}\mathbf{1}_{d_{\backslash\ell}} \\
    &\hspace{5pt} - \mathbf{\Psi}_\ell\mathrm{mat}_\ell\left[\mathcal{D}\right]\mathbf{1}_{d_{\backslash\ell}}- \sum_{\ell'\neq\ell}\mathrm{mat}_\ell\left[\mathcal{D} \times_{\ell'}\mathbf{\Psi}_{\ell'}\right]\mathbf{1}_{d_{\backslash\ell}}
\end{align*}

\subsection{Invertibility of $\mathbf{A}_\ell$}

Recall that, from the previous section, $\mathbf{A}_\ell$ can be written as:

\begin{align*}
    \mathbf{A}_\ell &= \left(\mathbf{I}_{d_\ell}\otimes\mathbf{1}_{d_{\backslash\ell}}\right)^T\mathbf{\Omega}\left(\mathbf{I}_{d_\ell}\otimes\mathbf{1}_{d_{\backslash\ell}}\right)
\end{align*}

$\mathbf{\Omega}_\ell$ is a precision matrix, i.e. positive definite.  Furthermore, $\left(\mathbf{I}_{d_\ell}\otimes\mathbf{1}_{d_{\backslash\ell}}\right)$ never maps nonzero vectors to the zero vector (it is a matrix that `copy-and-pastes' the input vector $d_{\backslash\ell}$ times).

\begin{align*}
    \mathbf{v}^T\mathbf{A}_\ell\mathbf{v} &= \mathbf{v}^T\left(\mathbf{I}_{d_\ell}\otimes\mathbf{1}_{d_{\backslash\ell}}\right)^T\mathbf{\Omega}\left(\mathbf{I}_{d_\ell}\otimes\mathbf{1}_{d_{\backslash\ell}}\right)\mathbf{v} & \mathbf{v} \neq \mathbf{0} \\
    &= \mathbf{w}^T\mathbf{\Omega}\mathbf{w} & \mathbf{w} \neq \mathbf{0} \\
    &> 0 \tag{positive-definiteness of $\mathbf{\Omega}$}
\end{align*}

In fact, this argument shows that $\mathbf{A}_\ell$ is not merely invertible, but positive definite!

What if $\mathbf{\Omega}$ is merely positive semidefinite?  Note that our $\mathrm{NLL}$ contains a $-\log |\mathbf{\Omega}|$ term, which diverges to infinity as $\mathbf{\Omega}$ approaches semidefiniteness.  Since it is known that optimizers for $\mathbf{\Omega}$ converge (see for example Theorem 6 in the paper that introduced TeraLasso \cite{greenewald_tensor_2019}), this cannot happen: $\mathbf{\Omega}$ must be positive definite.

Some work has been done investigating the case where we force $\mathbf{\Omega}$ to be restricted to the set of rank-$k$ matrices, such as in the preprint by \citeauthor{andrew_making_2024} (\citeyear{andrew_making_2024}).  The log-determinant becomes a log-pseudodeterminant.  In this case, $\mathbf{A}_\ell$ is indeed not invertible; the QP problem has multiple solutions.  However, in such work the nullspace of $\mathbf{\Omega}$ is known\footnote{The work finds an exact solution for the eigenvectors of $\mathbf{\Psi}_\ell$.}: we can add linear constraints to our QP problem to force our domain to be orthogonal to the nullspace, recovering uniqueness-of-coordinate-wise-solution and hence guaranteeing that the coordinate descent of Algorithm \ref{alg:mean-estim} converges.

\subsection{Proof of Global Optimality}

In this section, we will prove Theorem \ref{thm:unimodality} in the main paper.

\begin{repeatedtheorem}
    The maximum likelihood of the noncentral Kronecker-sum-structured normal distribution has a unique maximum, which is global.  The estimator defined by Algorithm \ref{alg:mean-estim} converges to this.
\end{repeatedtheorem}

The fact that Algorithm \ref{alg:mean-estim} converges follows from the fact that coordinate descent converges to stationary points if the coordinate-wise minima are unique, which we know is true given that our estimators for $\boldsymbol{\mu}_\ell, m, $ and $\mathbf{\Omega}$ are unique (estimating $\mathbf{\Omega}$ is not considered here, but it is a strictly convex optimization problem so whatever drop-in-algorithm we use for it should give a unique answer, potentially differing in the precise parameterization of $\mathbf{\Omega}$ but not $\mathbf{\Omega}$ itself).

To prove that the negative log-likelihood has a unique minimum, we intend to prove the following chain of implications:

\begin{align*}
    &\left(\boldsymbol{\omega}, \mathbf{\Omega}\right) \in \localmin{\mathfrak{A}} \\
    \underset{1}{\iff}& \left(\boldsymbol{\omega}, \mathbf{\Omega}, s=1\right) \in \localmin{\mathfrak{A}'} \\
    \underset{2}{\iff}& \left(\mathbf{\Gamma}\right) \in \localmin{\mathfrak{B}} \\
    \underset{3}{\iff}& \left(\mathbf{\Gamma}\right) \in \globmin{\mathfrak{B}}
\end{align*}

We will first focus on the case where $\mathbf{\Omega}$ is estimated with regularizers, rather than the low-rank subspace restriction.  Letting $\rho$ be some convex regularizer for the precision matrix (typically the graphical lasso, $\rho(\mathbf{\Omega}) = \sum_\ell \rho_\ell\lVert\mathbf{\Psi}_\ell\rVert_{1, od}$), we define our problems as:

\begin{align*}
    \mathfrak{A} &= \left\{\begin{matrix}
        \underset{\boldsymbol{\omega}, \mathbf{\Omega}}{\mathrm{argmin}} &\hspace{-2pt} -\log |\mathbf{\Omega}| + (\mathbf{x} - \boldsymbol{\omega})^T\mathbf{\Omega}(\mathbf{x} - \boldsymbol{\omega}) + \rho\left(\mathbf{\Omega}\right) \\
        \mathrm{where} & \mathbf{\Omega} \in \mathbb{K}_{\mathbf{M}} \succeq \mathbf{0} \\
        & \boldsymbol{\omega} \in \mathbb{K}_{\boldsymbol{v}}
    \end{matrix}\right.
    \\
    \mathfrak{A}' &= \left\{\begin{matrix}
        \underset{\boldsymbol{\omega}, \mathbf{\Omega}, s}{\mathrm{argmin}} &\hspace{-2pt} -\log |\mathbf{\Gamma}| + \begin{bmatrix}
        \boldsymbol{x}^T & 1
    \end{bmatrix}\mathbf{\Gamma}\begin{bmatrix}
        \boldsymbol{x} \\ 1
    \end{bmatrix} + \rho\left(\mathbf{\Omega}\right) \\
        \mathrm{where} & \mathbf{\Omega} \in \mathbb{K}_{\mathbf{M}} \succeq \mathbf{0} \\
        & \boldsymbol{\omega} \in \mathbb{K}_{\boldsymbol{v}} \\
        & s > 0 \\
        & \mathbf{\Gamma} = \begin{bmatrix}
        \mathbf{\Omega} & -\mathbf{\Omega}\boldsymbol{\omega} \\
        -\boldsymbol{\omega}^T\mathbf{\Omega} & \frac{1}{s} + \boldsymbol{\omega}^T\mathbf{\Omega}\boldsymbol{\omega}
    \end{bmatrix}
    \end{matrix} \right.
    \\
    \mathfrak{B} &= \left\{\begin{matrix}
        \underset{\mathbf{\Gamma}}{\mathrm{argmin}} &\hspace{-5pt} -\log |\mathbf{\Gamma}| + \begin{bmatrix}
        \boldsymbol{x}^T &\hspace{-7pt}1
    \end{bmatrix}\mathbf{\Gamma}\begin{bmatrix}
        \boldsymbol{x} \\ 1
    \end{bmatrix}\hspace{-1pt} + \rho\left(\mathbf{\Gamma}_{1:d_\forall,1:d_\forall}\right) \\
        \mathrm{where} & \mathbf{\Gamma}_{1:d_\forall,1:d_\forall} \in \mathbb{K}_{\mathbf{M}} \\
        & \mathbf{\Gamma}_{1:d_\forall,d_\forall+1} \in \mathbb{K}_{\mathbf{M}\boldsymbol{v}} \\
        & \mathbf{\Gamma} \succeq \mathbf{0}
    \end{matrix}\right.
\end{align*}

Recall that $\mathbb{K}_{\mathbf{M}\boldsymbol{v}}$ is a linear subspace, so $\mathfrak{B}$ corresponds to a conic program with linear constraints.  The objective is strictly convex.  Thus, $\mathcal{B}$ has a unique local minimum, which is also the global minimum.  This satisfies implication (3).

Next, note that there is a one-to-one correspondence between $(\boldsymbol{\omega}, \boldsymbol{\Omega}, s)$ and $\mathbf{\Gamma}$.  $\mathbf{\Omega} = \mathbf{\Gamma}_{1:d_\forall, 1:d_\forall}$, and since $\mathbf{\Omega}$ is guaranteed-to-be-invertible we can get $\boldsymbol{\omega}$ from $\mathbf{\Gamma}_{1:d_\forall,d_\forall+1}$.  Finally, $\frac{1}{s}$ can be derived with knowledge of the previous two parameters from $\mathbf{\Gamma}_{d_\forall+1,d_\forall+1}$.  Since the correspondence is one-to-one, the minima of $\mathfrak{A}'$ and $\mathfrak{B}$ are identical - since $\mathfrak{B}$ has a unique minimum, so does $\mathfrak{A}'$.  Thus, to prove implication (2) we need only show that the minimum of $\mathfrak{A}'$ occurs when $s=1$.

\begin{align*}
    \log|\mathbf{\Gamma}| &= \log\left(|\mathbf{\Omega}| \left|\frac{1}{s} + \boldsymbol{\omega}^T\mathbf{\Omega}\boldsymbol{\omega} - \boldsymbol{\omega}^T\mathbf{\Omega}\mathbf{\Omega}^{-1}\mathbf{\Omega}\boldsymbol{\omega}\right|\right) \tag{Block matrix determinant} \\
    &= \log|\mathbf{\Omega}| - \log s \\
    \begin{bmatrix}
        \boldsymbol{x}^T\hspace{-8pt}& 1
    \end{bmatrix}\hspace{-2pt}\mathbf{\Gamma}\hspace{-2pt}\begin{bmatrix}
        \boldsymbol{x} \\ 1
    \end{bmatrix} &= \begin{bmatrix}
        \mathbf{x}^T\mathbf{\Omega} - \boldsymbol{\omega}^T\mathbf{\Omega}\hspace{-3pt}&  \frac{1}{s} + \boldsymbol{\omega}^T\mathbf{\Omega}\boldsymbol{\omega} -\mathbf{x}^T\mathbf{\Omega}\boldsymbol{\omega}
    \end{bmatrix}\hspace{-4pt}\begin{bmatrix}
        \mathbf{x} \\ 1
    \end{bmatrix} \\
    &= \mathbf{x}^T\boldsymbol{\Omega}\mathbf{x} - \boldsymbol{\omega}^T\mathbf{\Omega}\mathbf{x} + \frac{1}{s} + \boldsymbol{\omega}^T\mathbf{\Omega}\boldsymbol{\omega} - \mathbf{x}^T\mathbf{\Omega}\boldsymbol{\omega} \\
    &= \left(\mathbf{x} - \boldsymbol{\omega}\right)^T\boldsymbol{\Omega}\left(\mathbf{x} - \boldsymbol{\omega}\right) + \frac{1}{s}
\end{align*}

Firstly, note that all terms involving $s$ are separable from the rest.  Differentiating with respect to $s$, we see that $0 = \frac{1}{s} - \frac{1}{s^2}$ at the minimum, i.e. that $s = s^2$; since $s > 0$, this leaves $s = 1$ must hold at the minimum; thus, implication (2) holds.  Secondly, note that these equations lead to the same $\mathrm{NLL}$ as $\mathfrak{A}$, except with an additive $\frac{1}{s} + \log s$ term thrown in.  Since we know $s=1$ at the minimum, $\mathfrak{A}'$ reduces to $\mathfrak{A}$; implication (1) holds as well.

This completes the proof; all three implications hold.  The log-likelihood of $\mathfrak{A}'$ and $\mathfrak{B}$ at $s=1$ are the same, and hence all local minima are global minima.  This holds when using regularization, but what about restrictions to a low-rank subspace?

\subsubsection{Global Optimality under Low-Rank Restrictions}

As mentioned in previous sections, most methods of estimating $\mathbf{\Omega}$ use regularization to guarantee existence of a solution.  Some restrict $\mathbf{\Omega}$ to a low-rank subspace, which requires changing determinants to pseudodeterminants.  All the implications still hold in this case, but the block-determinant formula no longer holds and the $\mathfrak{A}' \rightarrow \mathfrak{B}$ mapping is no longer one-to-one, since it is through invertibility of $\mathbf{\Omega}$ that we recover $\boldsymbol{\omega}$.  Under the low rank assumption, we can add the following constraint to $\mathfrak{A}$ and $\mathfrak{A}'$ to recover uniqueness.

\begin{align*}
    \boldsymbol{\omega} \perp \mathrm{nullspace}\left[\mathbf{\Omega}\right]
\end{align*}

The constraints work by enforcing $\boldsymbol{\omega}$ to be orthogonal to the nullspace of $\mathbf{\Omega}$.  Note that this has no effect on $\mathfrak{B}$, since any component of $\boldsymbol{\omega}$ that is not orthogonal to the nullspace gets mapped to zero.  This preserves the one-to-one mapping needed for implication (2), and the convexity needed for implication (3).

What needs to be re-proven for implication (1) to hold is the following:

\begin{enumerate}
    \item Local minima of $\mathfrak{A}'$ occur only at $s=1$
    \item At $s=1$, $\mathfrak{A}$ and $\mathfrak{A}'$ have the same solutions.
\end{enumerate}

Unfortunately, pseudodeterminants are \textit{much} harder to work with than typical determinants.  We rely on two results from the literature, from \citeauthor{holbrook_differentiating_2018} (\citeyear{holbrook_differentiating_2018}) and \citeauthor{hartwig_singular_1976} (\citeyear{hartwig_singular_1976}):

\begin{customthm}{2.20}[\citeauthor{holbrook_differentiating_2018} \citeyear{holbrook_differentiating_2018}]
\begin{align*}
    d\hspace{2pt} \mathrm{det}^\dagger \mathbf{A} &= \mathrm{det}^\dagger \left[\mathbf{A}\right]\mathrm{tr}\left[\mathbf{A}^\dagger d\mathbf{A}\right]
\end{align*}
Where $\mathrm{det}^\dagger$ is the pseudodeterminant and $\mathbf{A}^\dagger$ is the pseudoinverse.
\end{customthm}
\begin{corollary}
    \label{cor:holbrook}
    \begin{align*}
        \frac{d \hspace{2pt} \log \mathrm{det}^\dagger \mathbf{\Gamma}}{ds} &= \mathrm{tr}\left[\mathbf{\Gamma}^\dagger \frac{d\mathbf{\Gamma}}{ds}\right]
    \end{align*}
\end{corollary}

The second theorem, by \citeauthor{hartwig_singular_1976}, gives an expression for the pseudoinverse of a `bordered matrix', i.e. a matrix of the form $\begin{bmatrix}
        \mathbf{A} & \mathbf{c} \\
        \mathbf{b} & d
\end{bmatrix}$.  However, they split the problem into several sub-cases, which are not immediately apparent from the form of $\mathbf{A}, \mathbf{b}, \mathbf{c},$ and $d$.  Thus, we will first manipulate our matrix into one of these sub-cases before stating the theorem.

Recall that $\mathbf{\Omega}$ is positive semi-definite, and thus it has an eigendecomposition $\begin{bmatrix}
    \mathbf{V}_1 & \mathbf{V}_2
\end{bmatrix}\begin{bmatrix}
    \mathbf{\Lambda} & \mathbf{0} \\ \mathbf{0} & \mathbf{0}
\end{bmatrix}\begin{bmatrix}
    \mathbf{V}_1^T \\ \mathbf{V}_2^T
\end{bmatrix}$.  Letting $\mathbf{V} = \begin{bmatrix}
    \mathbf{V}_1 & \mathbf{V}_2
\end{bmatrix}$, observe the following:

\begin{align*}
    \begin{bmatrix}
        \mathbf{V}^T & \mathbf{0} \\
        \mathbf{0} & 1
    \end{bmatrix}&\begin{bmatrix}
        \mathbf{\Omega} & -\mathbf{\Omega}\boldsymbol{\omega} \\
        -\boldsymbol{\omega}^T\mathbf{\Omega} & \frac{1}{s} + \boldsymbol{\omega}^T\mathbf{\Omega}\boldsymbol{\omega}
    \end{bmatrix}\begin{bmatrix}
        \mathbf{V} & \mathbf{0} \\
        \mathbf{0} & 1
    \end{bmatrix} \\
    &= \begin{bmatrix}
        \mathbf{V}^T\mathbf{\Omega}\mathbf{V} & -\mathbf{V}^T\mathbf{\Omega}\boldsymbol{\omega} \\
        \boldsymbol{\omega}^T\mathbf{\Omega}\mathbf{V} & \frac{1}{s} + \boldsymbol{\omega}^T\mathbf{\Omega}\boldsymbol{\omega}
    \end{bmatrix} \\
    &= \begin{bmatrix}
        \mathbf{\Lambda} & \mathbf{0} & -\mathbf{\Lambda}\mathbf{V}^T_1\boldsymbol{\omega} \\
        \mathbf{0} & \mathbf{0} & \mathbf{0} \\
        -\boldsymbol{\omega}^T\mathbf{V}_1\mathbf{\Lambda} & \mathbf{0} & \frac{1}{s} + \boldsymbol{\omega}^T\mathbf{\Omega}\boldsymbol{\omega}
    \end{bmatrix}
\end{align*}

We will now find the pseudoinverse of this new matrix, as $\mathbf{U}\mathbf{M}^\dagger\mathbf{U}^T = \left(\mathbf{U}\mathbf{M}\mathbf{U}^T\right)^\dagger$ for unitary matrices $\mathbf{U}$.  In fact, the bottom-right entry is entirely unchanged by this transformation, and it will turn out to be all we need for our proof.

\begin{customthm}{\hspace{1pt}}[\citeauthor{hartwig_singular_1976} (\citeyear{hartwig_singular_1976})]
    Suppose we have a matrix of the form $\begin{bmatrix}
        \mathbf{\Sigma} & \mathbf{0} & \mathbf{q} \\
        \mathbf{0} & \mathbf{0} & \mathbf{0} \\
        \mathbf{p}^T & \mathbf{0} & d
    \end{bmatrix}$ for diagonal $\mathbf{\Sigma}$ of full rank, and $z = d - \mathbf{p}^T\mathbf{\Sigma}^{-1}\mathbf{q} \neq 0$.  Then, the lower right hand entry of the pseudoinverse is $\frac{1}{z}$.  Furthermore, the matrix has rank $1+\mathrm{rank}\mathbf{\Sigma}$.
\end{customthm}

This is, of course, a special case of their full result giving a formula for the entire pseudoinverse - but it is all we need.

\begin{corollary}
    \label{cor:hartwig}
    The lower right hand entry of $\mathbf{\Gamma}^\dagger$ is $s$.
\end{corollary}
\begin{proof}
    We already saw that our matrix can be turned into the form required by Hartwig's theorem without changing what the lower right hand entry will be.  It then suffices to observe that $z = \frac{1}{s} + \boldsymbol{\omega}^T\mathbf{\Omega}\boldsymbol{\omega} - \boldsymbol{\omega}^T\mathbf{V}_1\mathbf{\Lambda}\mathbf{\Lambda}^{-1}\mathbf{\Lambda}\mathbf{V}_1^T\boldsymbol{\omega} = \frac{1}{s}$.
\end{proof}

From Corollary \ref{cor:holbrook}, the derivative of the log pseudodeterminant is $\mathrm{tr}\left[\mathbf{\Gamma}^\dagger\begin{bmatrix}
    \mathbf{0} & \mathbf{0} \\
    \mathbf{0} & -\frac{1}{s^2}
\end{bmatrix}\right] = \mathrm{tr}\begin{bmatrix}
    \mathbf{0} & \frac{-1}{s^2}\mathbf{\Gamma}^\dagger_{1:d_\forall, d_\forall + 1} \\
    \mathbf{0} & \frac{-1}{s^2}\mathbf{\Gamma}^\dagger_{d_\forall+1,d_\forall+1}
\end{bmatrix}$.  Thus, it only depends on the bottom right entry of the pseudoinverse; in particular, it is $\frac{1}{s}$.  From this, the arguments in the previous section hold, and we can show that any minima occur only at $s = 1$.

All that remains is to show that at $s=1$, $\mathfrak{A}$ and $\mathfrak{A}'$ have the same solutions.  We already know from earlier that the pseudodeterminant of $\mathbf{\Gamma}$ is the same as the pseudodeterminant of $\begin{bmatrix}
        \mathbf{\Lambda} & \mathbf{0} & -\mathbf{\Lambda}\mathbf{V}^T_1\boldsymbol{\omega} \\
        \mathbf{0} & \mathbf{0} & \mathbf{0} \\
        -\boldsymbol{\omega}^T\mathbf{V}_1\mathbf{\Lambda} & \mathbf{0} & 1 + \boldsymbol{\omega}^T\mathbf{\Omega}\boldsymbol{\omega}
\end{bmatrix}$.  Removing out the zero rows and columns, we are left with $\begin{bmatrix}
        \mathbf{\Lambda} & -\mathbf{\Lambda}\mathbf{V}^T_1\boldsymbol{\omega} \\
        -\boldsymbol{\omega}^T\mathbf{V}_1\mathbf{\Lambda} & 1 + \boldsymbol{\omega}^T\mathbf{\Omega}\boldsymbol{\omega}
\end{bmatrix}$, which by Hartwig is a full rank matrix (recall the theorem stated its rank as $1 + \mathrm{rank}\mathbf{\Lambda}$ and that $\mathbf{\Lambda}$ corresponds to the non-zero eigenvalues).  Hence, we can use the ordinary block determinant rule.

\begin{align*}
    \mathrm{det}^\dagger \mathbf{\Gamma} &= \mathrm{det} \begin{bmatrix}
        \mathbf{\Lambda} & -\mathbf{\Lambda}\mathbf{V}^T_1\boldsymbol{\omega} \\
        -\boldsymbol{\omega}^T\mathbf{V}_1\mathbf{\Lambda} & 1 + \boldsymbol{\omega}^T\mathbf{\Omega}\boldsymbol{\omega}
    \end{bmatrix} \\
    &= \mathrm{det}\left[\mathbf{\Lambda}\right]\mathrm{det}\left[1 + \boldsymbol{\omega}^T\mathbf{\Omega}\boldsymbol{\omega} - \boldsymbol{\omega}^T\mathbf{V}_1\mathbf{\Lambda}\mathbf{V}_1^T\boldsymbol{\omega}\right]
\end{align*}

Recall that we required $\boldsymbol{\omega}$ to be orthogonal to the nullspace of $\mathbf{\Omega}$, resulting in:

\begin{align*}
    \mathrm{det}^\dagger \mathbf{\Gamma} &= \mathrm{det}\left[\mathbf{\Lambda}\right]\mathrm{det}\left[1 + \boldsymbol{\omega}^T\mathbf{V}_1\mathbf{\Lambda}\mathbf{V}_1^T\boldsymbol{\omega} - \boldsymbol{\omega}^T\mathbf{V}_1\mathbf{\Lambda}\mathbf{V}_1^T\boldsymbol{\omega}\right] \\
    &= \mathrm{det}\mathbf{\Lambda} \\
    &= \mathrm{det}^\dagger\mathbf{\Omega}
\end{align*}

Thus, we have that:

\begin{align*}
    \mathfrak{A}' &= \left\{\begin{matrix}
        \underset{\boldsymbol{\omega}, \mathbf{\Omega}, s}{\mathrm{argmin}} &\hspace{-2pt} -\log \det^\dagger\mathbf{\Gamma} + \begin{bmatrix}
        \boldsymbol{x}^T & 1
    \end{bmatrix}\mathbf{\Gamma}\begin{bmatrix}
        \boldsymbol{x} \\ 1
    \end{bmatrix} + \rho\left(\mathbf{\Omega}\right) \\
        \mathrm{where} & \mathbf{\Omega} \in \mathbb{K}_{\mathbf{M}} \succeq \mathbf{0} \\
        & \boldsymbol{\omega} \in \mathbb{K}_{\boldsymbol{v}}\\
        & \boldsymbol{\omega} \hspace{3pt}\bot\hspace{3pt} \mathrm{nullspace}\hspace{3pt}\mathbf{\Omega} \\
        & s > 0
\end{matrix}\right.
\end{align*}

is equivalent to:

\begin{align*}
    \mathfrak{A} &= \left\{\begin{matrix}
        \underset{\boldsymbol{\omega}, \mathbf{\Omega}}{\mathrm{argmin}} &\hspace{-8pt} -\log \det^\dagger\mathbf{\Omega} + \left(\mathbf{x}\hspace{-2pt}-\hspace{-2pt}\boldsymbol{\omega}\right)^T\mathbf{\Omega}\left(\mathbf{x}\hspace{-2pt}-\hspace{-2pt}\boldsymbol{\omega}\right) + \rho\left(\mathbf{\Omega}\right) \\
        \mathrm{where} & \mathbf{\Omega} \in \mathbb{K}_{\mathbf{M}} \succeq \mathbf{0} \\
        & \boldsymbol{\omega} \in \mathbb{K}_{\boldsymbol{v}}\\
        & \boldsymbol{\omega} \hspace{3pt}\bot\hspace{3pt} \mathrm{nullspace}\hspace{3pt}\mathbf{\Omega}
\end{matrix}\right.
\end{align*}

This completes the proof.

\section{Erdos-Renyi Synthetic Data}

We mentioned in the results section that we repeated our experiments with Erdos-Renyi graphs (p=0.05) instead of Barabasi-Albert graphs.  We found identical results, except that zero-mean algorithms perform much worse under Poisson noise.    See Figure \ref{fig:er} for the results, which use the same mean distributions as from the main paper.

\begin{figure}[h!]
    \centering
    \begin{subfigure}[t]{0.45\linewidth}
        \centering
        \includegraphics[width=\linewidth]{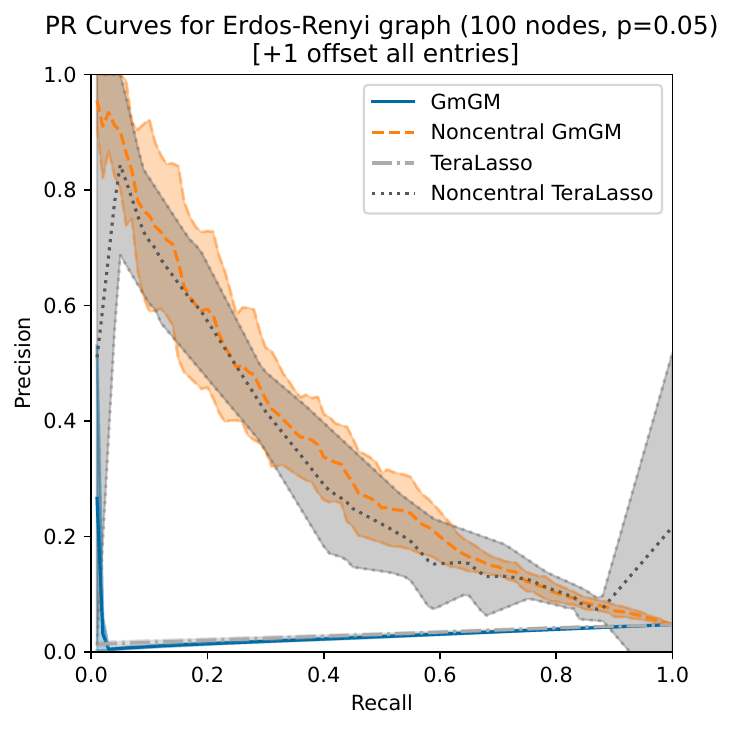}
    \end{subfigure}
    \begin{subfigure}[t]{0.45\linewidth}
        \centering
        \includegraphics[width=\linewidth]{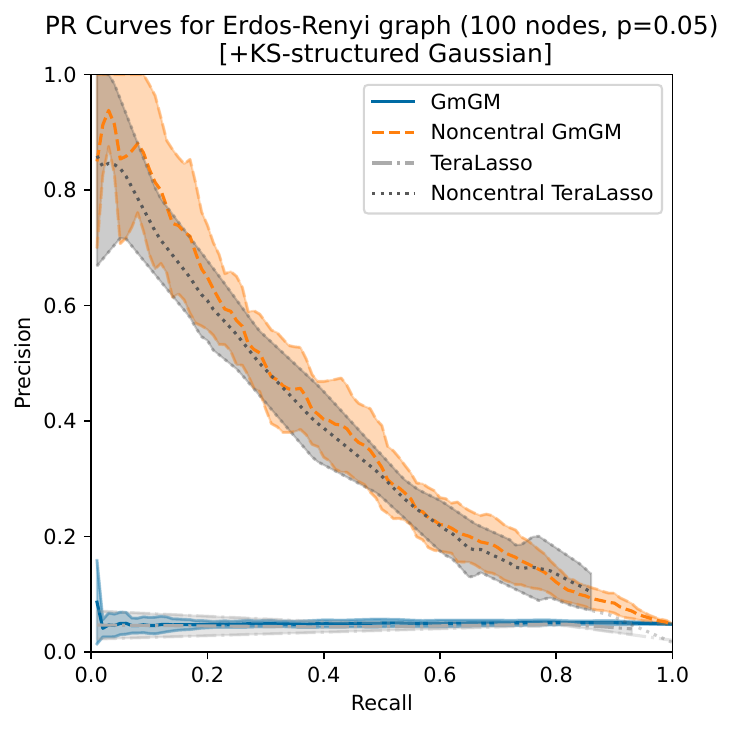}
    \end{subfigure}
    \begin{subfigure}[t]{0.45\linewidth}
        \centering
        \includegraphics[width=\linewidth]{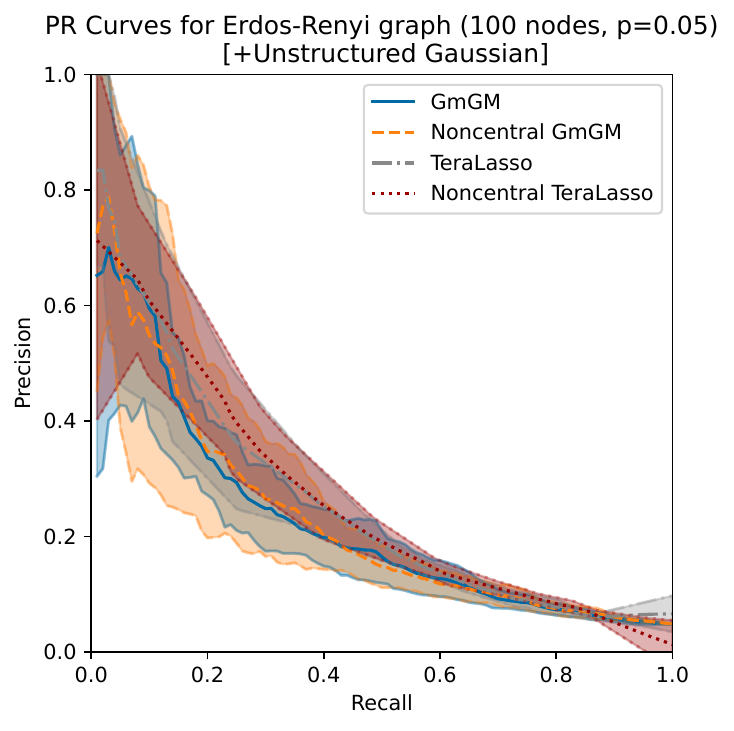}
    \end{subfigure}
    \begin{subfigure}[t]{0.45\linewidth}
        \centering
        \includegraphics[width=\linewidth]{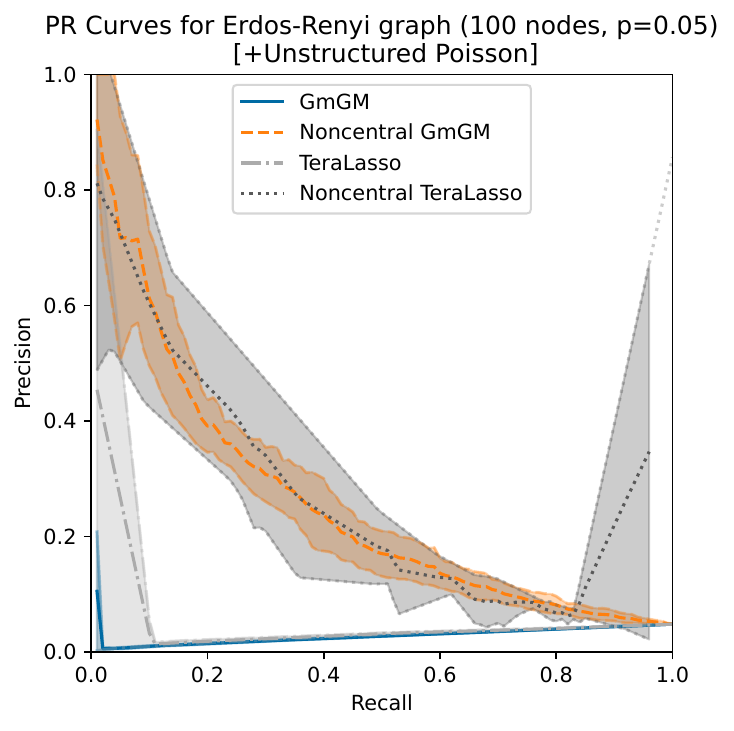}
    \end{subfigure}
    \caption{Precision and recall for on a synthetic dataset.  Error bars are the best/worst performance over 10 trials; the center line is average performance.}
    \label{fig:er}
\end{figure}

\end{document}